\documentclass[a4paper,10pt]{article}
\usepackage{geometry}
\usepackage[utf8]{inputenc}
\usepackage[T1]{fontenc}
\usepackage{graphicx}
\usepackage{caption}
\usepackage{subcaption}
\usepackage{amsmath, amssymb, amstext, mathtools}
\usepackage{amsthm}
\usepackage{tikz}
\usepackage{pgfplots}
\usepackage{float}
\usepackage{hyperref}

\usetikzlibrary{patterns,arrows,decorations.pathreplacing,decorations.pathmorphing}

\theoremstyle{definition}
\newtheorem{theorem}{Theorem}[section]
\newtheorem{proposition}[theorem]{Proposition}
\newtheorem{lemma}[theorem]{Lemma}
\newtheorem{definition}[theorem]{Definition}
\newtheorem{corollary}[theorem]{Corollary}

\newtheorem{remark}[theorem]{Remark}

\DeclareMathOperator{\Aut}{Aut}
\DeclareMathOperator{\Sym}{Sym}
\DeclareMathOperator{\Stab}{Stab}

\DeclareMathOperator{\im}{im}

\DeclareMathOperator{\argmin}{argmin}
\DeclareMathOperator{\argmax}{argmax}

\DeclareMathOperator{\ori}{ori}

\numberwithin{equation}{section}
\numberwithin{figure}{section}

\title{Graph Isomorphism for unit square graphs}
\author{Daniel Neuen\\RWTH Aachen University\\\texttt{neuen@informatik.rwth-aachen.de}}

\date{}

\begin{document}

\maketitle

\begin{abstract}
 In the past decades for more and more graph classes the Graph Isomorphism Problem was shown to be solvable in polynomial time.
 An interesting family of graph classes arises from intersection graphs of geometric objects.
 In this work we show that the Graph Isomorphism Problem for unit square graphs, intersection graphs of axis-parallel unit squares in the plane, can be solved in polynomial time.
 Since the recognition problem for this class of graphs is NP-hard we can not rely on standard techniques for geometric graphs based on constructing a canonical realization.
 Instead, we develop new techniques which combine structural insights into the class of unit square graphs with understanding of the automorphism group of such graphs.
 For the latter we introduce a generalization of bounded degree graphs which is used to capture the main structure of unit square graphs.
 Using group theoretic algorithms we obtain sufficient information to solve the isomorphism problem for unit square graphs.
\end{abstract}

\section{Introduction}

The Graph Isomorphism Problem is one of the most famous open problems in theoretical computer science.
In the past three decades the problem was intensively studied but only recently the upper bound on the complexity could be improved to quasipolynomial time \cite{babai15}.
However, it is still open whether the Graph Isomorphism Problem can be solved in polynomial time.
In this work we focus on geometric graph classes, that is, graph classes that arise as intersection graphs of geometric objects.
In an intersection graph the vertices are identified with geometric objects and two vertices are connected if the corresponding objects intersect.

One of the most basic geometric graph classes is the class of interval graphs, intersection graphs of intervals on the real line.
Although this graph class is quite restrictive there are a number of practical applications and specialized algorithms for interval graphs (see e.g.\ \cite{gol04}).
The Graph Isomorphism Problem on interval graphs can be solved in linear time \cite{cb81} as well as in logarithmic space \cite{kklv11}.
However, for several generalizations of interval graphs the complexity of the Graph Isomorphism Problem is unknown.
This includes for example circular arc graphs (see \cite{clmnsss13}) and triangle graphs (see \cite{ueh14}).
On the other hand a graph class is GI-complete if the Graph Isomorphism Problem for this class is as difficult as the general problem under polynomial time reductions.
An example of a GI-complete geometric class is the class of grid intersection graphs, bipartite intersection graphs of horizontal and vertical line segments in the plane \cite{ueh08}.
As an immediate consequence the class of intersection graphs of axis-parallel rectangles is also GI-complete.
Unit square graphs, intersection graphs of axis-parallel unit squares, are a natural restriction for the rectangle graphs.
This raises the question for the complexity of the isomorphism problem for unit square graphs.
In this work we prove that the Graph Isomorphism Problem for unit square graphs can be solved in polynomial time.
Besides being a natural restriction to rectangle graphs, another central motivation to study this problem comes from unit disk graphs, intersection graphs of unit circles in the plane.
Unit disk graphs where first studied by Clark, Colbourn and Johnson in \cite{ccj90} and for several problems specialized algorithms have been proposed (see e.g.\ \cite{FLS12}).
Practical applications arise for example from broadcast networks where each broadcast station is represented by a vertex and two stations communicate with each other if the distance between them does not exceed the broadcast range.
In the work of Clark et al.\ two problems, namely the recognition problem and the isomorphism problem, were left open.
While recognition of unit disk graphs proved to be NP-hard \cite{bk98}, the isomorphism problem for unit disk graphs is still open.
Unit square graphs present a natural variant to unit disk graphs as we just replace the Euclidean norm by the Manhattan norm.
Also, going from unit disks to unit squares removes geometric intricacies and tends to simplify the structure of graphs but maintains several key aspects of the problem.
In particular, for unit disk as well as unit square graphs vertices only have a bounded number of independent neighbors (set of pairwise non-adjacent neighbors) and the structure of graphs seems to be a mixture of bounded degree and planarity. 
Hence, solving the isomorphism problem for unit square graphs might be a step towards solving the same problem for unit disk graphs.
Furthermore, in \cite{ueh08, ueh14} Uehara also asked for the complexity of graph isomorphism for unit grid intersection graphs.
Unit grid intersection graphs can be seen as bipartite versions of unit square graphs in the following sense:
A bipartite graph is a unit grid intersection graph if and only if it is the intersection graph of unit squares where intersections between squares belonging to vertices on the same side of the bipartition are ignored.
Hence, the result presented in this work shows that in some sense the difficulty for unit grid intersection lies in recreating the information which lines are close to each other.

Another interesting point arises from the fact that, like for unit disk graphs, recognition of unit square graphs is NP-hard \cite{breu1996}.
Hence, we obtain an example of a natural graph class with the interesting property that isomorphism tests can be performed in polynomial time whereas recognition is NP-hard.
Also, the hardness result rules out the classical approach to attack the isomorphism problem.
Typically, isomorphism tests for geometric graphs are based on constructing a canonical geometric representation of the graph (see e.g.\ \cite{kklv11, kkv13}) but, as this would also solve the recognition problem, such an approach is not possible here.
Instead, our algorithm combines group theoretic techniques with geometric properties of unit square graphs.
For the group theoretic machinery we extend the results developed by Luks \cite{luks82} to decide isomorphism for bounded degree graphs by also allowing for example large cycles in the neighborhood of a vertex.
On a geometric level this coincides in some sense with the intuition that vertices in the neighborhood of some fixed vertex are cyclically arranged around the central vertex.
Using geometric properties of unit square graphs and known algorithms for other geometric graph classes such as proper circular arc graphs we can canonically (in an isomorphism-invariant way) extract such circular orderings, which can then be used by the group theoretic machinery.
For this, we show a series of results giving a deep insight into the structure of neighborhoods of single vertices and neighborhoods of cliques within unit square graphs.
These results not only help us to understand the structure of unit square graphs, but also we obtain significant knowledge about the structure of the automorphism group of a unit square graph.
However, an obvious obstacle to this approach comes from large cliques which are connected in a uniform way to the rest of the graph and do not contain any significant structure.
More precisely, such cliques may be responsible for large symmetric groups which are subgroups of the automorphism group of the whole graph.
Since large symmetric groups form a clear obstacle to the group theoretic machinery and can not be handled by Luks' algorithm we have to cope with these parts of the graph in a different way.
Our second main result on the structure of unit square graphs characterizes connections between cliques which are stable with respect to the color refinement algorithm (see e.g.\ \cite{bbg13, mckay81}), and also establishes a close connection to interval graphs.
Building on this characterization we show that the color refinement algorithm can be used to cope with the symmetric parts of the graph containing no significant structure.
Finally, combining the color refinement algorithm with the group theoretic machinery, we obtain an algorithm to solve the isomorphism problem for unit square graphs.

The remainder of this paper is structured as follows.
In the next section we give some preliminaries and in Section \ref{sec:basic-properties} we prove some basic properties of unit square graphs.
In Section \ref{sec:local} we analyze the local structure.
In Sections \ref{sec:neighborhoods} and \ref{sec:global} we compute the desired partitions together with the canonical structure.
Finally we discuss some related open problems in Section \ref{sec:discussion}.

\section{Preliminaries}

\subsection{Graphs}

A \emph{graph} is a pair $G = (V, E)$ with vertex set $V = V(G)$ and edge set $E = E(G)$.
In this paper all graphs are undirected, so $E(G)$ is always irreflexive and symmetric.
The \emph{(open) neighborhood} of a vertex $v \in V(G)$ is the set $N_G(v)  = N(v) = \{w \in V(G) \mid vw \in E(G)\}$ and the size of $N(v)$ is the degree of $v$.
The \emph{closed neighborhood} is the set $N[v] = N(v) \cup \{v\}$.
Two vertices $v,w \in V(G)$ are \emph{connected twins} if $N[v] = N[w]$.
A \emph{path} from $v$ to $w$ of length $m$ is a sequence $u_0,\dots,u_m$ of distinct vertices with $u_0 = v$ and $u_m = w$, such that $u_{i-1}u_{i} \in E(G)$ for each $i \in [m] := \{1,\dots,m\}$.
The \emph{distance} between $v$ and $w$, $d(v,w)$, is the length of a shortest path from $v$ to $w$.
A colored graph is a tuple $G = (V,E,c)$ where $c\colon V(G) \rightarrow \mathbb{N}$ assigns each vertex a unique color.
For each color $i \in \mathbb{N}$ let $V_i(G) = \{v \in V(G) \mid c(v) = i\}$.
Two graphs $G$ and $H$ are \emph{isomorphic} ($G \cong H$) if there is a bijection $\varphi \colon V(G) \rightarrow V(H)$, such that $vw \in E(G)$ if and only if $\varphi(v)\varphi(w) \in E(H)$ for all $v,w \in V(G)$.
In this case the mapping $\varphi$ is an \emph{isomorphism} from $G$ to $H$.
In case the input graphs are colored it is demanded that the isomorphism also preserves the coloring of the vertices.
The Graph Isomorphism Problem asks whether two given graphs $G$ and $H$ are isomorphic.
An isomorphism from a graph to itself is called an \emph{automorphism}.
The set of automorphisms of a graph $G$, denoted by $\Aut(G)$, forms a subgroup of the symmetric group over the vertex set.

\subsection{Color Refinement}

A very basic and fundamental method, which is a basic building block of many isomorphism tests, is the color refinement algorithm (see e.g.\;\cite{mckay81}).
The basic idea is to iteratively distinguish vertices if they have a different number of neighbors in some color.
A partition $\mathcal{P}$ of the vertices is \emph{stable} if for all $X,Y \in \mathcal{P}$ and all $v,w \in X$ it holds that $|N(v) \cap Y| = |N(w) \cap Y|$.
Further a partition $\mathcal{P}$ \emph{refines} another partition $\mathcal{Q}$ if for each $X \in \mathcal{P}$ there is some $Y \in \mathcal{Q}$ with $X \subseteq Y$.
The color refinement algorithm computes the unique coarsest stable partition refining the initial color partition (i.e.\ the partition of the vertices according to their color).
The coarsest stable partition can be computed in almost linear time (see \cite{mckay81,bbg13}).
We say color refinement distinguishes two graphs if there is some class in the coarsest stable partition on the disjoint union of the graphs that contains a different number of vertices from the two graphs.
In this case the two input graphs are not isomorphic.

The $k$-dimensional Weisfeiler-Leman algorithm is a generalization of the color refinement algorithm (cf.\ \cite{cfi92}).
Instead of coloring only single vertices, the Weisfeiler-Leman algorithm colors $k$-tuples of vertices.
Initially each tuple is colored with the isomorphism type of the underlying induced subgraph and then the coloring is refined in a similar way as by the color refinement algorithm (see \cite{cfi92} for a detailed description).
We say $k$-dimensional Weisfeiler-Leman \emph{identifies} a graph class $\mathcal{C}$ if for every pair of non-isomorphic graphs $G,H$ with $G \in \mathcal{C}$ the $k$-dimensional Weisfeiler-Leman distinguishes between $G$ and $H$.

\subsection{Group Theory}

\subsubsection{Basic Facts}

In this subsection we introduce some group theoretic notation.
For a general introduction to group theory we refer to \cite{rotman} whereas several group theoretic algorithms are given in \cite{cgt,seress}.
Since we only deal with automorphism groups of graphs we restrict ourselves to permutation groups.

Let $\Omega$ be a finite set and $n = |\Omega|$.
The \emph{symmetric group} $\Sym(\Omega)$ over the set $\Omega$ is the group containing all permutations on $\Omega$.
A \emph{permutation group} over the set $\Omega$ is a subgroup of the group $\Sym(\Omega)$.
Let $\Gamma \leq \Sym(\Omega)$ be a permutation group.
For $\gamma \in \Gamma$ and $\alpha \in \Omega$ we denote by $\alpha^{\gamma}$ the image of $\alpha$ under the permutation $\gamma$.
The set $\alpha^{\Gamma} = \{\alpha^{\gamma} \mid \gamma \in \Gamma\}$ is the \emph{orbit} of $\alpha$.
The group $\Gamma$ is \emph{transitive} if $\alpha^{\Gamma} = \Omega$ for some (and thus for all) $\alpha \in \Omega$.
The \emph{stabilizer} of $\alpha$ is $\Stab_{\Gamma}(\alpha) = \{\gamma \in \Gamma \mid \alpha^{\gamma} = \alpha\}$ which is a subgroup of $\Gamma$.
For $A \subseteq \Omega$ let $A^{\gamma} = \{\alpha^{\gamma} \mid \alpha \in A\}$.
The \emph{(setwise) stabilizer} of $A$ is $\Stab_{\Gamma}(A) = \{\gamma \in \Gamma \mid A^{\gamma} = A\}$.
The set $A$ is called \emph{$\Gamma$-invariant} if $\Gamma = \Stab_{\Gamma}(A)$.
For a partition $\mathcal{P}$ of the set $\Omega$ let $\mathcal{P}^{\gamma} = \{X^{\gamma} \mid X \in \mathcal{P}\}$ which is again a partition of $\Omega$.
The partition $\mathcal{P}$ is \emph{$\Gamma$-invariant} if $\mathcal{P}^{\gamma} = \mathcal{P}$ for all $\gamma \in \Gamma$.
The group $\Gamma$ is abelian if $\gamma\gamma' = \gamma'\gamma$ for all $\gamma,\gamma' \in \Gamma$.
A subgroup $N \leq \Gamma$ is \emph{normal} ($N \unlhd \Gamma$) if $\gamma^{-1}N\gamma = \{\gamma^{-1}n\gamma \mid n \in N\} = N$ for all $\gamma \in \Gamma$.
A mapping $\varphi\colon \Gamma \to \Delta$ is a \emph{homomorphism} if $\varphi(\gamma_1)\varphi(\gamma_2) = \varphi(\gamma_1\gamma_2)$ for all $\gamma_1,\gamma_2 \in \Gamma$.
Its kernel is $\ker(\varphi) = \{\gamma \in \Gamma \mid \varphi(\gamma) = 1\} \unlhd \Gamma$ and the image is $\im(\varphi) = \{\varphi(\gamma) \mid \gamma \in \Gamma\} \leq \Delta$.

A set $S \subseteq \Gamma$ is a \emph{generating set} if each $\gamma \in \Gamma$ can be written as $\gamma = s_1 \dots s_k$ for some $s_1,\dots,s_k \in S$.
In order to allow efficient computations permutation groups are represented by small generating sets.
For each permutation group on a set of size $n$ there is a generating set of size $n-1$ \cite{mn87}.
In practice generating sets of quadratic size are typically used to represent permutation groups (cf.\ \cite{fhl80}).

\subsubsection{Bounded Composition Series}

In this work we shall be interested in a particular subclass of permutation groups, namely groups with bounded non-abelian composition factors.
Let $\Gamma$ be a group.
A \emph{normal series} is a sequence of subgroups $\Gamma = \Lambda_0 \trianglerighteq \Lambda_1 \trianglerighteq \dots \trianglerighteq \Lambda_k = \{1\}$.
The length of the series is $k$ and the groups $\Lambda_{i-1} / \Lambda_{i}$ are the factor groups of the series, $i \in [k]$.
A \emph{composition series} is a strictly decreasing normal series of maximal length. 
For every finite group $\Gamma$ all composition series have the same family of factor groups considered as a multi-set (cf.\ \cite{rotman}).
A \emph{composition factor} of a finite group $\Gamma$ is a factor group of a composition series of $\Gamma$.

\begin{definition}
 Let $d \in \mathbb{N}$.
 The family $\Gamma_d$ contains all finite groups $\Gamma$ for which all non-abelian composition factors are isomorphic to subgroups of $S_d = \Sym([d])$.
\end{definition}

The class of $\Gamma_d$-groups is closed under subgroups and homomorphic images.
Furthermore, for groups $N \trianglelefteq \Gamma$ it holds that $\Gamma \in \Gamma_d$ if and only if $N \in \Gamma_d$ and $\Gamma/N \in \Gamma_d$.
A group $\Gamma$ is \emph{solvable} if every composition factor is abelian.
Two examples of solvable groups, that are particularly important for this work, are cyclic groups and dihedral groups which are the automorphism groups of directed cycles and undirected cycles.
Note that every solvable group is a $\Gamma_d$-group for every $d \in \mathbb{N}$.

\subsubsection{Algorithms}

The Setwise Stabilizer Problem asks, given a permutation group $\Gamma \leq \Sym(\Omega)$ and $A \subseteq \Omega$, for a generating set of the group $\Stab_{\Gamma}(A)$.
A central motivation to consider $\Gamma_d$-groups is the following result.

\begin{proposition}
 \label{prop:setwise-stab-gamma-d}
 Let $d \in \mathbb{N}$.
 The Setwise Stabilizer Problem for groups in $\Gamma_d$ can be solved in polynomial time.
\end{proposition}

A weaker version of this statement was proved by Luks in \cite{luks82} considering only groups where all composition factors are isomorphic to subgroups of $S_d$.
For the more general version stated above we refer to \cite{babai14}.
This result is for example used by Luks to solve graph isomorphism for graphs of bounded degree \cite{luks82}, but it can also be applied to more general graph classes.
In this work we are interested in graphs which we call \emph{$t$-circle-bounded} graphs.
For a graph $G$ and a set $X \subseteq V(G)$ we write $G[X]$ to denote the induced subgraph of $G$ with vertex set $X$.

\begin{definition}
 A colored graph $G = (V,E,c)$ with $c\colon V(G) \rightarrow [k]$ is \emph{$t$-circle-bounded} if for each $i \in [k]$ and $X \subseteq V_{< i} := \bigcup_{j < i} V_j(G)$ the graph
 \begin{equation*}
  G_{i,X} := G[\{v \in V_i(G) \mid N(v) \cap V_{< i} = X\}]
 \end{equation*}
 is the disjoint union of at most $t$ connected graphs of maximum degree two.
\end{definition}

The $t$-circle-bounded graphs are closely related to $t$-bounded graphs which are similarly defined with the size of $G_{i,X}$ bounded by $t$ (see \cite{BCSTW13}).
From an algorithmic point of view we can use very similar methods for the isomorphism problem on $t$-circle-bounded graphs as for $t$-bounded graphs.
For the sake of completeness we still give a complete description.
This requires the following result due to Miller on hypergraphs which is also used in the algorithm for $t$-bounded graphs.
A \emph{hypergraph} is a pair $\mathcal{H} = (V,\mathcal{E})$ with finite vertex set $V$ and a set of hyperedges $\mathcal{E} \subseteq 2^{V}$.

\begin{proposition}[Miller, \cite{miller83}]
 \label{prop:hypergraph-gamma-t}
 Let $t \in \mathbb{N}$.
 Let $\mathcal{H} = (V, \mathcal{E})$ be a hypergraph and $\Gamma \leq \Sym(V)$, such that $\Gamma \in \Gamma_t$.
 Then a generating set for the group $\Aut(\mathcal{H}) \cap \Gamma$ can be computed in polynomial time.
\end{proposition}

Note, that this result also holds for edge-colored hypergraphs by considering one color after the other.

\begin{theorem}
 \label{thm:gamma-t-t-bounded}
 Let $G$ be a $t$-circle-bounded graph.
 Then $\Aut(G) \in \Gamma_t$.
\end{theorem}

\begin{proof}
 We proof by induction on $i \in [k]$ that $\Aut(G[V_{\leq i}]) \in \Gamma_t$.
 For $i=1$ the graph $G[V_1]$ is the disjoint union of at most $t$ connected graphs of maximum degree two.
 A connected graph of maximum degree two is either a path or a cycle.
 In both cases the automorphism group is solvable and thus it is a $\Gamma_t$-group.
 Let $X_1,\dots,X_s$ be the connected components of $G[V_1]$.
 Then $s \leq t$ and $\mathcal{P}_1 = \{X_1,\dots,X_s\}$ is an $\Aut(G[V_1])$-invariant partition.
 Consider the natural action $\psi\colon\Aut(G[V_1]) \rightarrow \Sym(\mathcal{P}_1)$ of $\Aut(G[V_1])$ on the connected components.
 Then $\ker(\psi) \in \Gamma_t$, since it is a direct product of solvable groups, and $\im(\psi) \leq S_t$.
 Hence, $\Aut(G[V_1]) \in \Gamma_t$.
 
 For the inductive step suppose $\Aut(G[V_{<i}]) \in \Gamma_t$.
 Consider the restriction $\varphi_i\colon \Aut(G[V_{\leq i}]) \to \Aut(G[V_{<i}])\colon \gamma \mapsto \gamma|_{V_{<i}}$.
 Then $\im(\varphi_i) \leq \Aut(G[V_{<i}])$ and hence, $\im(\varphi_i) \in \Gamma_t$ by induction hypothesis.
 So it remains to prove that $\ker(\varphi_i) \in \Gamma_t$.
 For each $X \subseteq V_{<i}$ the set $V[G_{i,X}]$ is $\ker(\varphi_i)$-invariant.
 So \[\ker(\varphi_i) \leq \bigtimes_{X \subseteq V_{<i}\colon V[G_{i,X}] \neq \emptyset} \Aut(G_{i,X}).\]
 Since a subgroup of a direct product of $\Gamma_t$-groups is still a $\Gamma_t$-group it suffices to prove that $\Aut(G_{i,X}) \in \Gamma_t$.
 But this follows by the same argument as in the base step of the induction.
\end{proof}

\begin{theorem}
 \label{thm:automorphism-t-bounded}
 The Graph Isomorphism Problem for $t$-circle-bounded graphs can be solved in polynomial time.
\end{theorem}

\begin{proof}
 By standard reduction techniques it suffices to show that a generating set for the automorphism group of a $t$-circle-bounded graph can be computed in polynomial time.
 
 Let $G$ be a $t$-circle-bounded graph.
 We proof by induction on $i \in [k]$ that a generating set for $\Aut(G[V_{\leq i}])$ can be computed in polynomial time.
 For $i=1$ the graph $G[V_1]$ has maximum degree two and thus, a generating set can be computed in polynomial time.
 So let $i > 1$ and suppose we are given a generating set for the group $\Aut(G[V_{< i}])$.
 Let $\mathcal{E}_i = \{X \subseteq V_{<i} \mid V(G_{i,X}) \neq \emptyset\}$ and define the hypergraph $\mathcal{H}_i = (V_{<i},\mathcal{E}_i)$.
 Let $G_i = (V_{\leq i}, E(G_i))$ with
 \begin{equation*}
  E(G_i) = E(G[V_{<i}]) \cup \{vw \in E(G) \mid \exists X \in \mathcal{E}_i : vw \in E(G_{i,X})\}.
 \end{equation*}
 Then $\Aut(G[V_{\leq i}]) \leq \Aut(G_i)$ since $G_i$ is defined in an isomorphism-invariant way.
 Consider the homomorphism $\varphi_i\colon \Aut(G_i) \to \Aut(G[V_{<i}])\colon \gamma \mapsto \gamma|_{V_{<i}}$.
 Then $\ker(\varphi_i) = \bigtimes_{X \in \mathcal{E}_i} \Aut(G_{i,X})$ which can be computed in polynomial time.
 Further $\ker(\varphi_i) \in \Gamma_t$ by Theorem \ref{thm:gamma-t-t-bounded}.
 Let $K$ be a generating set for the kernel $\ker(\varphi_i)$.
 Now let $X \sim Y$ if $G_{i,X} \cong G_{i,Y}$ for all $X,Y \in \mathcal{E}_i$ and define a coloring $c \colon \mathcal{E}_i \to \mathbb{N}$, so that each color class corresponds to an equivalence class of $\sim$.
 Then $\im(\varphi_i) = \Aut(\mathcal{H},c) \cap \Aut(G[V_{<i}])$.
 Let $S$ be a generating set for $\im(\varphi_i)$ which can be computed in polynomial time by Proposition \ref{prop:hypergraph-gamma-t}.
 For $s \in S$ compute some $s' \in \varphi_i^{-1}(s)$ using isomorphisms between $G_{i,X}$ and $G_{i,X^{s}}$ for each $X \in \mathcal{E}_i$.
 Let $S'$ be the set of all elements $s'$.
 Then $K \cup S'$ is a generating set for $\Aut(G_i)$.
 Furthermore $\Aut(G_i) \in \Gamma_t$.
 So a generating set for $\Aut(G[V_{\leq i}])$ can be computed in polynomial time using Proposition \ref{prop:setwise-stab-gamma-d}.
\end{proof}

\subsection{Proper circular arc graphs}

A graph $G$ is a \emph{unit interval graph} if $G$ is the intersection graph of unit intervals on the real line.
A graph $G$ is a \emph{proper circular arc graph} if $G$ is the intersection graph of arcs on a circle, such that for no two arcs, one is properly contained in the other.
A characterization of unit interval graphs and proper circular arc graphs in terms of forbidden induced subgraphs is given in \cite{tuck74}.
For our purposes the following statements are sufficient.
Some relevant forbidden induced subgraphs are also depicted in Figure \ref{fig:pca-forbidden}.

\begin{proposition}
 A graph $G$ is a unit interval graph if and only if there are no induced subgraphs isomorphic to $C_{n+4}$ for $n \geq 0$, $S_3$, $K_{1,3}$ and \texttt{net}.
\end{proposition}

Here, $C_{n}$ denotes a cycle of length $n$.
Furthermore we denote by $G \cup H$ the disjoint union of $G$ and $H$ and the graph $\overline{G}$ is the complement graph of $G$.

\begin{proposition}
 \label{prop:pca-forbidden-characterization}
 Let $G$ be a graph, such that $N_G[v]$ is a unit interval graph for every $v \in V(G)$ and there are no induced subgraphs isomorphic to $K_1 \cup C_{n+4}$ for $n \geq 0$, $K_1 \cup S_3$, $\overline{T_2}$, $\overline{C_6}$ and \texttt{net}.
 Then $G$ is a proper circular arc graph.
\end{proposition}

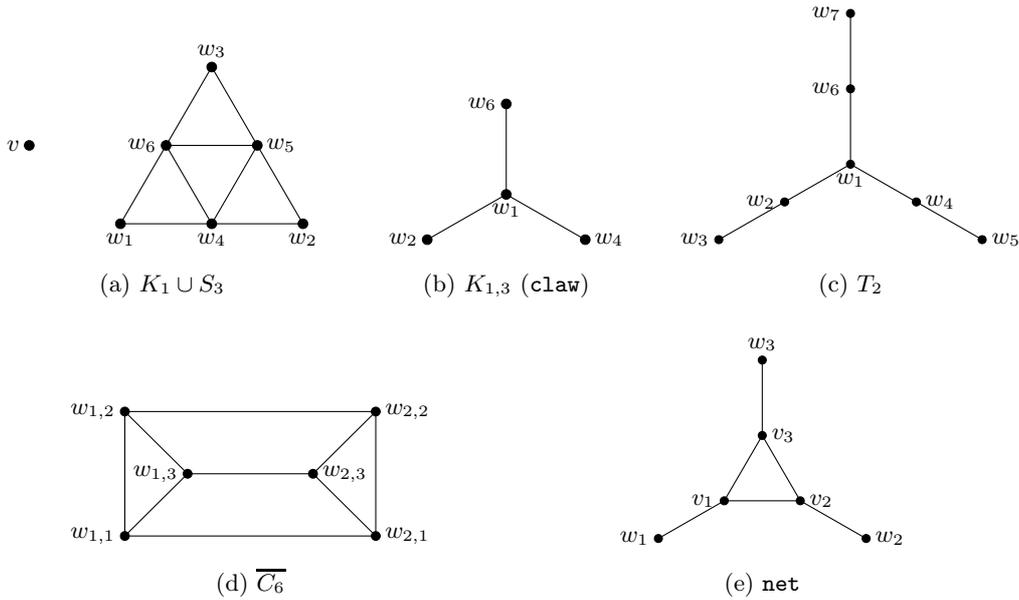
\begin{figure}
 \centering
 \begin{subfigure}[b]{.35\linewidth}
  \centering
  \begin{tikzpicture}[scale=1.2]
   \draw[fill] (0,0.866) circle (1.5pt) node[left] {{\small $v$}};
   \draw[fill] (1,0) circle (1.5pt) node[below] {{\small $w_1$}};
   \draw[fill] (2,0) circle (1.5pt) node[below] {{\small $w_4$}};
   \draw[fill] (3,0) circle (1.5pt) node[below] {{\small $w_2$}};
   \draw[fill] (1.5,0.866) circle (1.5pt) node[left] {{\small $w_6$}};
   \draw[fill] (2.5,0.866) circle (1.5pt) node[right] {{\small $w_5$}};
   \draw[fill] (2,1.732) circle (1.5pt) node[above] {{\small $w_3$}};
   
   \draw
    (1,0) -- (2,0)
    (2,0) -- (3,0)
    (1,0) -- (1.5,0.866)
    (2,0) -- (1.5,0.866)
    (2,0) -- (2.5,0.866)
    (3,0) -- (2.5,0.866)
    (1.5,0.866) -- (2.5,0.866)
    (1.5,0.866) -- (2,1.732)
    (2.5,0.866) -- (2,1.732);
  \end{tikzpicture}
  \caption{$K_1 \cup S_3$}
  \label{fig:graph-s3}
 \end{subfigure}
 \begin{subfigure}[b]{.25\linewidth}
  \centering
  \begin{tikzpicture}[scale=1.2]
   \draw[fill] (2,2) circle (1.5pt) node[below] {{\small $w_1$}};
   \draw[fill] (2,3) circle (1.5pt) node[left] {{\small $w_6$}};
   \draw[fill] (1.134,1.5) circle (1.5pt) node[left] {{\small $w_2$}};
   \draw[fill] (2.866,1.5) circle (1.5pt) node[right] {{\small $w_4$}};
   
   \draw
    (2,2) -- (2,3)
    (2,2) -- (1.134,1.5)
    (2,2) -- (2.866,1.5);
  \end{tikzpicture}
  \caption{$K_{1,3}$ (\texttt{claw})}
  \label{fig:graph-k13}
 \end{subfigure}
 \begin{subfigure}[b]{.35\linewidth}
  \centering
  \begin{tikzpicture}
   \draw[fill] (2,2) circle (1.5pt) node[below] {{\small $w_1$}};
   \draw[fill] (2,3) circle (1.5pt) node[left] {{\small $w_6$}};
   \draw[fill] (2,4) circle (1.5pt) node[left] {{\small $w_7$}};
   \draw[fill] (1.134,1.5) circle (1.5pt) node[left] {{\small $w_2$}};
   \draw[fill] (2.866,1.5) circle (1.5pt) node[right] {{\small $w_4$}};
   \draw[fill] (0.268,1) circle (1.5pt) node[left] {{\small $w_3$}};
   \draw[fill] (3.732,1) circle (1.5pt) node[right] {{\small $w_5$}};
   
   \draw
    (2,2) -- (2,3)
    (2,2) -- (1.134,1.5)
    (2,2) -- (2.866,1.5)
    (2,3) -- (2,4)
    (1.134,1.5) -- (0.268,1)
    (2.866,1.5) -- (3.732,1);
  \end{tikzpicture}
  \caption{$T_2$}
  \label{fig:graph-t2}
 \end{subfigure}
 \begin{subfigure}[b]{.45\linewidth}
  \centering
  \vspace{10pt}
  \begin{tikzpicture}[scale=1.1]
   \draw[fill] (0,0) circle (1.5pt) node[left] {{\small $w_{1,1}$}};
   \draw[fill] (0,1.5) circle (1.5pt) node[left] {{\small $w_{1,2}$}};
   \draw[fill] (0.75,0.75) circle (1.5pt) node[left] {{\small $w_{1,3}$}};
   \draw[fill] (3,0) circle (1.5pt) node[right] {{\small $w_{2,1}$}};
   \draw[fill] (3,1.5) circle (1.5pt) node[right] {{\small $w_{2,2}$}};
   \draw[fill] (2.25,0.75) circle (1.5pt) node[right] {{\small $w_{2,3}$}};
   
   \draw
    (0,0) -- (0,1.5)
    (0,0) -- (0.75,0.75)
    (0.75,0.75) -- (0,1.5)
    (0,0) -- (3,0)
    (0.75,0.75) -- (2.25,0.75)
    (0,1.5) -- (3,1.5)
    (3,0) -- (3,1.5)
    (3,0) -- (2.25,0.75)
    (2.25,0.75) -- (3,1.5);
  \end{tikzpicture}
  \caption{$\overline{C_6}$}
  \label{fig:graph-c6}
 \end{subfigure}
 \begin{subfigure}[b]{.45\linewidth}
  \centering
  \vspace{10pt}
  \begin{tikzpicture}
   \draw[fill] (1,1) circle (1.5pt) node[left] {{\small $v_1$}};
   \draw[fill] (2,1) circle (1.5pt) node[right] {{\small $v_2$}};
   \draw[fill] (1.5,1.866) circle (1.5pt) node[right] {{\small $v_3$}};
   \draw[fill] (0.134,0.5) circle (1.5pt) node[left] {{\small $w_1$}};
   \draw[fill] (2.866,0.5) circle (1.5pt) node[right] {{\small $w_2$}};
   \draw[fill] (1.5,2.866) circle (1.5pt) node[above] {{\small $w_3$}};
   
   \draw
    (1,1) -- (2,1)
    (1,1) -- (1.5,1.866)
    (2,1) -- (1.5,1.866)
    (1,1) -- (0.134,0.5)
    (1.5,1.866) -- (1.5,2.866)
    (2,1) -- (2.866,0.5);
  \end{tikzpicture}
  \caption{\texttt{net}}
  \label{fig:graph-net}
 \end{subfigure}
 \caption{Some forbidden induced subgraphs of proper circular arc graphs}
 \label{fig:pca-forbidden}
\end{figure}

We also require the following characterization for proper circular arc graphs.
A vertex is \emph{universal} if it is adjacent to all other vertices.
A graph $G$ is \emph{twin-free} if it does not contain connected twins and $G$ is \emph{co-bipartite} if $\overline{G}$ is bipartite.

\begin{proposition}[\cite{kkv12}, Theorem 3]
 \label{prop:pca-by-cycle}
 Let $G$ be a graph without universal vertices.
 Then $G$ is a proper circular arc graph if and only if there is a cycle $H$ with $V(G) = V(H)$, such that
 \begin{enumerate}
  \item $N_G[v]$ induces a connected subgraph in $H$ for each $v \in V(G)$
  \item For all $v,w \in V(G)$ it holds that if $N_G[v] \subseteq N_G[w]$ then the two paths share an endpoint in $H$.
 \end{enumerate}
\end{proposition}

Furthermore, if $G$ is connected, twin-free and not co-bipartite, the cycle $H$ is unique \cite{huang95}.
Additionally, given some proper circular arc graph, a cycle $H$ can be computed in polynomial time (see \cite{kkv12}).
This gives us the following proposition.

\begin{proposition}
 \label{prop:pca-circular-order}
 Let $G$ be a connected proper circular arc graph, such that $\overline{G}$ is not bipartite.
 Further let $\sim_G$ be the connected twins relation and $\mathcal{P}$ the corresponding partition into equivalence classes.
 Then one can compute in polynomial time a canonical connected graph $H$, such that $V(H) = \mathcal{P}$ and $H$ has maximum degree two.
\end{proposition}

\section{Basic Properties}
\label{sec:basic-properties}

For unit square graphs there are several possible definitions.
The most obvious one is to describe vertices by axis-parallel unit squares with edges connecting two vertices if the unit squares intersect.
Alternatively it can also be demanded that vertices represented by unit squares are connected if the center of the first square is contained in the other square.
Another possibility is to describe vertices by points in the plane.
Note that two squares with unit side length intersect if and only if the distance between their centers using the maximum norm is at most one.
Thus, two vertices are connected if the distance between the points is at most one using the maximum norm.
Furthermore, a unit square contains the center point of another unit square if and only if the distance between both centers is at most one half using the maximum norm.
By applying a scaling argument this also gives us the equivalence to the second definition.
In this paper we work with the last definition, that is we represent vertices by points in the plane.
For a point $p \in \mathbb{R}^{k}$ we denote by $p_i$ the $i$-th component of $p$, $i \in [k]$.
The $L_{\infty}$-norm is defined as $\|p \|_{\infty} = \max_{i \in [k]}p_i$.

\begin{definition}
 A \emph{$k$-dimensional $L_{\infty}$-realization} of a graph $G$ is a mapping $f:V(G) \rightarrow \mathbb{R}^{k}$ such that $vw \in E(G)$ if and only if $\|f(v) - f(w)\|_{\infty} \leq 1$ for all $v,w \in V(G)$.
 A \emph{unit square graph} is a graph having a two-dimensional $L_{\infty}$-realization.
\end{definition}

Observe that graphs with $1$-dimensional $L_{\infty}$-realization are exactly the unit interval graphs.
For the remainder of this paper we focus on unit square graphs and just use the term realization for a two-dimensional $L_{\infty}$-realization.
Following our previous notation, for a realization $f:V(G) \rightarrow \mathbb{R}^{2}$ and a vertex $v \in V(G)$, we denote by $(f(v))_i$ the $i$-th component of $f(v)$.
This is also abbreviated by $f(v)_i$, $i \in [2]$.
We start by showing a series of basic properties for unit square graphs.

\begin{lemma}
 \label{la:forbid-k15}
 Let $G$ be a unit square graph. Then $G$ has no induced subgraph isomorphic to $K_{1,5}$.
\end{lemma}

\begin{proof}
 Since the class of unit square graphs is hereditary (closed under induced subgraphs) it suffices to show that $K_{1,5}$ is not a unit square graph.
 Suppose towards a contradiction that there exists a realization $f:V(G) \rightarrow \mathbb{R}^{2}$.
 Without loss of generality let $f(v) = (0,0)$ where $v$ is the center vertex connected to the other vertices $w_1,\dots,w_5$.
 Then $f(w_i) \in [-1,1]^{2}$ for all $i \in [5]$.
 Thus, there is some $q_1,q_2 \in \{-1,1\}$, such that there are $i \neq j$ with $f(w_i),f(w_j) \in [q_1,0] \times [q_2,0]$ (for simplicity we may use $[1,0]$ to denote the interval $[0,1]$).
 But then there is an edge between $w_i$ and $w_j$ which is a contradiction.
\end{proof}

\begin{remark}
 \label{re:induced-unit-interal-by-realization}
 Let $G$ be a unit square graph and $f\colon V(G) \rightarrow \mathbb{R}^{2}$ a realization.
 Further let $X \subseteq V(G)$, such that there are $a_1,b_1,a_2,b_2 \in \mathbb{R}$ with $a_1 \leq b_1 \leq a_1+1$, $a_2 \leq b_2$ and $f(v) \in [a_1,b_1] \times [a_2,b_2]$ for every $v \in X$.
 Then $G[X]$ is a unit interval graph.
\end{remark}

\begin{lemma}
 \label{la:unit-interval-neighborhood-existence}
 Let $G$ be a unit square graph. Then there is some $v \in V(G)$, such that $G[N[v]]$ is a unit interval graph.
\end{lemma}

\begin{proof}
 Let $f\colon V(G) \rightarrow \mathbb{R}^{2}$ be a realization and choose $v = \argmin_{v \in V(G)} f(v)_1$.
 Further let $a_1 = f(v)_1$, $b_1 = a_1 + 1$, $a_2 = f(v)_2 - 1$ and $b_2 = f(v)_2 + 1$.
 Then $f(w) \in [a_1,b_1] \times [a_2,b_2]$ for every $w \in N[v]$.
 So $G[N[v]]$ is a unit interval graph according to Remark \ref{re:induced-unit-interal-by-realization}.
\end{proof}

\begin{lemma}
 Let $u,v \in V(G)$ be two non-adjacent vertices. Then $N(u) \cap N(v)$ defines a unit interval graph with at most two independent vertices.
\end{lemma}

\begin{proof}
 Without loss of generality assume $f(u)_1 + 1 \leq f(v)_1$.
 For $w \in N(u) \cap N(v)$ we obtain $f(u)_1 \leq f(w)_1 \leq f(v)_1$.
 Without loss of generality let $f(u) = (0,0)$.
 Then $f(w) \in [0,1] \times [-1,1]$ for every $w \in N(u) \cap N(v)$. 
 So $N(u) \cap N(v)$ defines a unit interval graph according to Remark \ref{re:induced-unit-interal-by-realization}.
 Furthermore in this area there can be at most two independent vertices.
\end{proof}

\begin{corollary}
 \label{cor:forbidden-k23-3k2-t2}
 Let $G$ be a unit square graph. Then $G$ has no induced subgraph isomorphic to $K_{2,3}$, $\overline{3K_2}$ ($3K_2$ is the disjoint union of three $K_2$) or $\overline{T_2}$.
\end{corollary}

\begin{proof}
 For $K_{2,3}$ the two vertices on the left side have three independent common neighbors.
 The graph $\overline{3K_2}$ contains two non-adjacent vertices whose common neighborhood is a $4$-cycle.
 Finally $\overline{T_2}$ extends the graph $\overline{3K_2}$.
\end{proof}

We also require some properties of maximal cliques.
A clique is a set of vertices $C \subseteq V(G)$, such that $vw \in E(G)$ for every two distinct $v,w \in C$.
A maximal clique is a clique so that there is no larger clique containing it.
For a graph $G$ the set of maximal cliques of $G$ is denoted by $\mathcal{M}(G)$.

\begin{lemma}
 \label{la:clique-common-neighborhood}
 Let $G$ be a unit square graph and $C$ be a maximal clique of $G$.
 Then there are $v_1,\dots,v_4 \in V(G)$, such that $C = \bigcap_{i \in [4]} N[v_i]$.
\end{lemma}

\begin{proof}
 Let $f:V(G) \rightarrow \mathbb{R}^{2}$ be a realization of $G$.
 Pick $v_{2i-1} = \text{argmin}_{v \in C}f(v)_i$ and $v_{2i} = \text{argmax}_{v \in C}f(v)_i$ for $i \in [2]$. 
 Clearly $C \subseteq \bigcap_{i \in [4]} N[v_i]$.
 So let $w \in \bigcap_{i \in [4]} N[v_i]$ and $v \in C$.
 In order to prove $w \in C$ it suffices to show that $\| f(v) - f(w)\|_{\infty} \leq 1$, since $v$ is chosen arbitrarily from $C$.
 For $i \in [2]$ it holds that $f(v_{2i-1})_i \leq f(v_{2i})_i$.
 If $f(w)_i \leq f(v_{2i})_i$ then $-1 \leq f(v_{2i})_i - 1 - f(v)_i \leq f(w)_i - f(v)_i \leq f(v_{2i})_i - f(v)_i \leq 1$.
 Otherwise $f(w)_i \geq f(v_{2i-1})_i$ and $-1 \leq  f(v_{2i-1})_i - f(v)_i \leq f(w)_i - f(v)_i \leq f(v_{2i-1})_i + 1 - f(v)_i \leq 1$.
\end{proof}

The proof of the last lemma motivates us to describe maximal cliques by unit squares with all points being inside the square belonging to the maximal clique.

\begin{definition}
 Let $G$ be a unit square graph, $C \in \mathcal{M}(G)$ and let $f\colon V(G) \rightarrow \mathbb{R}^{2}$ be a realization.
 The \emph{center} of $C$ with respect to $f$ is the set $z_f(C) := \{p \in \mathbb{R}^{2}\;|\;\forall v \in C\colon \|f(v) - p\|_{\infty} \leq \frac{1}{2}\}$.
\end{definition}

Choose $v_1,\dots,v_4$ as in the last lemma and let $p_i = \frac{f(v_{2i-1})_i + f(v_{2i})_i}{2}$ for $i \in [2]$.
Then $p \in z_f(C)$ and so $z_f(C) \neq \emptyset$.
Further for every $p \in z_f(C)$ the unit square with center placed at $p$ contains exactly those points that belong to $C$.
This can easily be seen from the definition of the center and using the triangle inequation.
Thus, for two distinct maximal cliques $C,D$ we get $z_f(C) \cap z_f(D) = \emptyset$.

\begin{lemma}
 \label{la:clique-center-border}
 Let $G$ be a unit square graph and $C$ be a maximal clique of $G$.
 Further let $v_{2i-1} = \text{argmin}_{v \in C}f(v)_i$ and $v_{2i} = \text{argmax}_{v \in C}f(v)_i$ for $i \in [2]$.
 Then $z_f(C) = [f(v_2)_1 - \frac{1}{2}, f(v_1)_1 + \frac{1}{2}] \times [f(v_4)_2 - \frac{1}{2}, f(v_3)_2 + \frac{1}{2}]$.
\end{lemma}

\begin{proof}
 Clearly $z_f(C) \subseteq [f(v_2)_1 - \frac{1}{2}, f(v_1)_1 + \frac{1}{2}] \times [f(v_4)_2 - \frac{1}{2}, f(v_3)_2 + \frac{1}{2}]$.
 So let $p \in [f(v_2)_1 - \frac{1}{2}, f(v_1)_1 + \frac{1}{2}] \times [f(v_4)_2 - \frac{1}{2}, f(v_3)_2 + \frac{1}{2}]$ and $w \in C$.
 If $f(w)_1 \geq p_1$ then $0 \leq f(w)_1 - p_1 \leq f(v_2)_1 - p_1 \leq \frac{1}{2}$.
 Otherwise $f(w)_1 \leq p_1$ and hence, $0 \leq p_1 - f(w)_1 \leq p_1 - f(v_1)_1 \leq \frac{1}{2}$.
 The argument for the second component is analogous.
 Thus, $\|w - p\|_{\infty} \leq \frac{1}{2}$.
\end{proof}

\section{Local Structure}
\label{sec:local}

The basic approach for our algorithm is group-theoretic.
A main obstacle for group theoretic approaches comes from large symmetric or alternating groups, that appear in the automorphism group.
For unit square graphs the central observation is that these groups can in a way only arise from cliques.
In this section we show how to cope with possibly very symmetric parts of the graphs and give a corresponding reduction to get rid of them.

\subsection{Connection to Interval graphs}

In order to obtain a better understanding of how the symmetric parts may look like we start by giving a translation from interval graphs to unit square graphs.
This is based on the description of maximal cliques introduced before.
Let $C_1,\dots,C_k$ be maximal cliques of a unit square graph, such that for all $i < j$ and all $p^{i} \in z_f(C_i)$, $p^{j} \in z_f(C_j)$ it holds that $p^{i}_b < p^{j}_b$ for $b \in [2]$.
Then it is not difficult to show that each vertex can only appear in consecutive maximal cliques.
More precisely, for each interval $I \subseteq [k]$, there is a non-empty set of points $P \subseteq \mathbb{R}^{2}$, such that for each $v \in V(G)$ with $f(v) \in P$ it holds $v \in C_i$ if and only if $i \in I$.
A visualization is given in Figure \ref{fig:interval-to-unit-square-2}.
This allows us to encode arbitrary interval graphs using unit square graphs.

Let $G$ be a graph.
Define the colored graph $G_{\mathcal{M}} = (V(G) \cup \mathcal{M}(G), E(G_{\mathcal{M}}), c_{G_{\mathcal{M}}})$ with
\begin{equation*}
 E(G_{\mathcal{M}}) = \{vC \mid C \in \mathcal{M}(G), v \in C\} \cup \{vw \mid v \neq w \in V(G)\} \cup \{CD \mid C \neq D \in \mathcal{M}(G)\}
\end{equation*}
and $c_{G_{\mathcal{M}}}(v) = c_G(v) + 1$ for $v \in V(G)$, $c_{G_{\mathcal{M}}}(C) = 1$ for $C \in \mathcal{M}(G)$.
For a group $\Gamma \leq \Sym(\Omega)$ and a set $A \subseteq \Omega$ the \emph{pointwise stabilizer} is the group $\Stab_\Gamma^{\bullet}(A) := \{\gamma \in \Gamma \mid \forall \alpha \in A\colon\alpha^{\gamma} = \alpha\}$.
If $A$ is invariant under $\Gamma$ (i.e.\ $A^{\gamma} = A$ for every $\gamma \in \Gamma$) we define the restriction of $\Gamma$ to $A$ as $\Gamma|_A := \{\gamma|_A \mid \gamma \in \Gamma\}$ where $\gamma|_A \colon A \rightarrow A$ with $\gamma|_A(\alpha) = \gamma(\alpha)$.

\begin{lemma}
 \label{la:clique-graph-aut}
 For every two graphs $G$ and $H$ it holds that
 \begin{enumerate}
  \item $G \cong H$ if and only if $G_{\mathcal{M}} \cong H_{\mathcal{M}}$,
  \item $\Stab_{\Aut(G_{\mathcal{M}})}^{\bullet}(V(G)) = \{1\}$ (here $1$ denotes the identify mapping),
  \item $\Aut(G_{\mathcal{M}})|_{V(G)} = \Aut(G)$.
 \end{enumerate}
\end{lemma}

\begin{proof}
 This follows directly from the definition of $G_{\mathcal{M}}$ and the fact that $vw \in E(G)$ if and only if there is some maximal clique $C \in \mathcal{M}(G)$ with $v,w \in C$.
\end{proof}

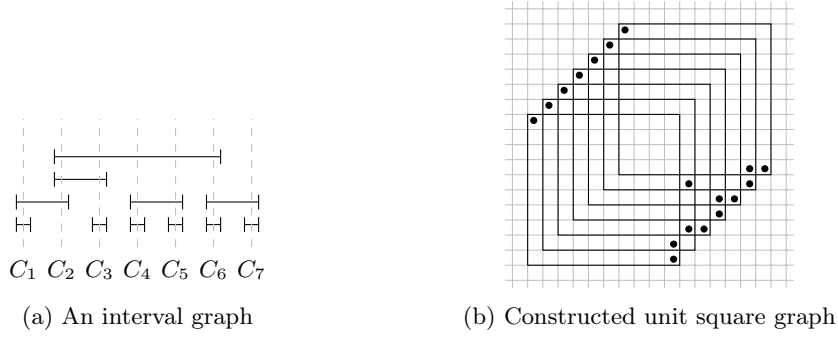
\begin{figure}
 \centering
 \begin{subfigure}[b]{.45\linewidth}
  \centering
  \begin{tikzpicture}
   \draw[|-|] (0.4,0.6) -- (0.6,0.6);
   \draw[|-|] (0.4,0.9) -- (1.1,0.9);
   \draw[|-|] (0.9,1.2) -- (1.6,1.2);
   \draw[|-|] (1.4,0.6) -- (1.6,0.6);
   \draw[|-|] (1.9,0.6) -- (2.1,0.6);
   \draw[|-|] (2.4,0.6) -- (2.6,0.6);
   \draw[|-|] (1.9,0.9) -- (2.6,0.9);
   \draw[|-|] (0.9,1.5) -- (3.1,1.5);
   \draw[|-|] (2.9,0.6) -- (3.1,0.6);
   \draw[|-|] (3.4,0.6) -- (3.6,0.6);
   \draw[|-|] (2.9,0.9) -- (3.6,0.9);
   
   \draw[dashed, gray!50] (0.5,0.3) -- (0.5,2);
   \draw[dashed, gray!50] (1.0,0.3) -- (1.0,2);
   \draw[dashed, gray!50] (1.5,0.3) -- (1.5,2);
   \draw[dashed, gray!50] (2.0,0.3) -- (2.0,2);
   \draw[dashed, gray!50] (2.5,0.3) -- (2.5,2);
   \draw[dashed, gray!50] (3.0,0.3) -- (3.0,2);
   \draw[dashed, gray!50] (3.5,0.3) -- (3.5,2);
   
   \node at (0.5,0) {{\small $C_1$}};
   \node at (1.0,0) {{\small $C_2$}};
   \node at (1.5,0) {{\small $C_3$}};
   \node at (2.0,0) {{\small $C_4$}};
   \node at (2.5,0) {{\small $C_5$}};
   \node at (3.0,0) {{\small $C_6$}};
   \node at (3.5,0) {{\small $C_7$}};
  \end{tikzpicture}
  \caption{An interval graph}
  \label{fig:interval-to-unit-square-1}
 \end{subfigure}
 \begin{subfigure}[b]{.45\linewidth}
  \centering
  \begin{tikzpicture}
   \draw[step=0.2, gray!50, very thin] (0.1,0.1) grid (3.9,3.9);
   
   \draw (0.4,0.4) -- (2.4,0.4) -- (2.4,2.4) -- (0.4,2.4) -- (0.4,0.4);
   \draw (0.6,0.6) -- (2.6,0.6) -- (2.6,2.6) -- (0.6,2.6) -- (0.6,0.6);
   \draw (0.8,0.8) -- (2.8,0.8) -- (2.8,2.8) -- (0.8,2.8) -- (0.8,0.8);
   \draw (1.0,1.0) -- (3.0,1.0) -- (3.0,3.0) -- (1.0,3.0) -- (1.0,1.0);
   \draw (1.2,1.2) -- (3.2,1.2) -- (3.2,3.2) -- (1.2,3.2) -- (1.2,1.2);
   \draw (1.4,1.4) -- (3.4,1.4) -- (3.4,3.4) -- (1.4,3.4) -- (1.4,1.4);
   \draw (1.6,1.6) -- (3.6,1.6) -- (3.6,3.6) -- (1.6,3.6) -- (1.6,1.6);
   
   \draw[fill] (0.48,2.32) circle (0.04);
   \draw[fill] (0.68,2.52) circle (0.04);
   \draw[fill] (0.88,2.72) circle (0.04);
   \draw[fill] (1.08,2.92) circle (0.04);
   \draw[fill] (1.28,3.12) circle (0.04);
   \draw[fill] (1.48,3.32) circle (0.04);
   \draw[fill] (1.68,3.52) circle (0.04);
   
   \draw[fill] (2.32,0.48) circle (0.04);
   \draw[fill] (2.32,0.68) circle (0.04);
   \draw[fill] (2.52,0.88) circle (0.04);
   \draw[fill] (2.72,0.88) circle (0.04);
   \draw[fill] (2.92,1.08) circle (0.04);
   \draw[fill] (2.92,1.28) circle (0.04);
   \draw[fill] (3.12,1.28) circle (0.04);
   \draw[fill] (2.52,1.48) circle (0.04);
   \draw[fill] (3.32,1.48) circle (0.04);
   \draw[fill] (3.32,1.68) circle (0.04);
   \draw[fill] (3.52,1.68) circle (0.04);
  \end{tikzpicture}
  \caption{Constructed unit square graph}
  \label{fig:interval-to-unit-square-2}
 \end{subfigure}
 \caption{From interval to unit square graphs}
 \label{fig:interval-to-unit-square}
\end{figure}

We now use the fact that a graph $G$ is an interval graph if and only if there is a linear order on the maximal cliques, such that each vertex appears in consecutive maximal cliques \cite{fg65}.
For a vertex $v \in V(G)$ let $\mathcal{M}_G(v) = \mathcal{M}(v) = \{C \in \mathcal{M}(G) \mid v \in C\}$.

\begin{lemma}
 \label{la:interval-to-unit-square}
 Let $G$ be an interval graph.
 Then $G_{\mathcal{M}}$ is a colored unit square graph.
\end{lemma}

\begin{proof}
 Let $<$ be a linear order on $\mathcal{M}(G)$, such that each vertex appears in consecutive maximal cliques.
 Let $k = |\mathcal{M}(G)|$ and $C_1 < C_2 < \dots < C_k$ be the maximal cliques of $G$.
 For each $v \in V(G)$ define $a_v, b_v \in [k]$ in such a way that $\mathcal{M}(v) = \{C_i \;|\; a_v \leq i \leq b_v\}$.
 
 Consider the following realization $f: V(G_{\mathcal{M}}) \rightarrow \mathbb{R}^{2}$ with
 \begin{equation*}
  f(C_i) = (\frac{i}{k} - 1, \frac{i}{k})
 \end{equation*}
 and
 \begin{equation*}
  f(v) = (\frac{a_v}{k}, \frac{b_v}{k} - 1)
 \end{equation*}
 for all $v \in V(G)$.
 Clearly all vertices are connected to each other as well as all maximal cliques.
 So let $v \in V(G)$.
 Then $|\frac{a_v}{k} - \frac{i}{k} + 1| = |\frac{a_v - i}{k} + 1| \leq 1$ if and only if $a_v \leq i$.
 Further $|\frac{i}{k} - \frac{b_v}{k} + 1| = |\frac{i - b_v}{k} + 1| \leq 1$ if and only if $i \leq b_v$.
 So there is an edge between $v \in V(G)$ and $C_i \in \mathcal{M}$ if and only if $a_v \leq i \leq b_v$ if and only if $v \in C_i$.
\end{proof}

A visualization of the last lemma is given in Figure \ref{fig:interval-to-unit-square}.

\begin{corollary}
 \label{cor:translation-interval-to-unit-square}
 For each colored interval graph $G$ one can compute in polynomial time some colored unit square graph $H$ with $\Aut(G) \cong \Aut(H)$.
\end{corollary}

\begin{proof}
 This follows directly from Lemma \ref{la:clique-graph-aut} and \ref{la:interval-to-unit-square}.
\end{proof}

In particular this construction implies that there are twin-free unit square graphs where the automorphism group contains arbitrarily large symmetric groups which can not be handled by the group theoretic approach due to Luks.
For example consider the following graph $G_{n,k}$ for $n,k \in \mathbb{N}$.
The vertex set $V(G_{n,k}) = [n]^{\leq k}$ is the set of all words over the alphabet $[n]$ of length at most $k$ and there is an edge $vw \in E(G_{n,k})$ if $v$ is a prefix of $w$ (this is interpreted for an undirected graph).
First, $G_{n,k}$ is an interval graph.
To verify this consider the set $[n]^{k}$ of words of length $k$ with the natural lexicographic order.
Then each vertex $v \in V(G_{n,k})$ can be represented by the interval $I_{n,k}(v) = \{w \in [n]^{k} \mid v \text{ is prefix of }w\}$.
It is easy to check that two vertices are connected if and only if the corresponding intervals intersect.
The automorphism group of $G_{n,k}$ is a wreath product of the automorphism group of $G_{n,k-1}$ by a symmetric group on $n$ points.

One of the main contributions of this work is to show that within automorphism groups of unit square graphs such symmetric groups only appear in a local setting.
Here, local refers to a small area in the realization of a graph $G$.
In the presented example the vertices of each color class are close together and in particular they induce a clique.
The main target for this section is to present a method how to cope with the local parts of the graph that may admit large symmetric groups in the automorphism group.
For this purpose we have to analyze the structure of clique-partitions of the vertices.

\subsection{Clique-Partitions}

\begin{definition}
 Let $G$ be a graph and $\mathcal{P}$ be a partition of the vertices.
 We call $\mathcal{P}$ a \emph{clique-partition} if $X$ is a clique for each $X \in \mathcal{P}$.
\end{definition}

In the following let $G$ be a unit square graph with realization $f:V(G) \rightarrow \mathbb{R}^{2}$.
We first define the graph $G_{\mathcal{M}}^{*} = (V(G) \cup \mathcal{M}(G), E(G_{\mathcal{M}}^{*}), c_{G_{\mathcal{M}}^{*}})$ with
\begin{equation*}
 E(G_{\mathcal{M}}^{*}) = \{vC \mid C \in \mathcal{M}(G), v \in C\} \cup E(G)
\end{equation*}
and $c_{G_{\mathcal{M}}^{*}} = c_{G_{\mathcal{M}}}$ as defined in the previous subsection.
Let $\mathcal{P}$ be a clique-partition of the vertices, which is refined by the color refinement algorithm applied to the graph $G_{\mathcal{M}}^{*}$.
More precisely let $\mathcal{P}^{*}$ be the coarsest partition of $V(G) \cup \mathcal{M}(G)$ that is stable with respect to $G_{\mathcal{M}}^{*}$ and refines the partition $\mathcal{P} \cup \{\mathcal{M}(G)\}$.
The partition $\mathcal{P}$ is called \emph{clique-stable} if $\mathcal{P}^{*} \cap 2^{V(G)} = \mathcal{P}$.
In following we analyze the structure of $\mathcal{P}^{*}$ with respect to the realization $f$.
For this purpose let $\simeq$ be the equivalence relation on $V(G) \cup \mathcal{M}(G)$ that corresponds to $\mathcal{P}^{*}$.

\begin{lemma}
 Let $C,D$ be distinct maximal cliques with $C \simeq D$. Then $z_f(C)_i \cap z_f(D)_i = \emptyset$ for $i \in [2]$.
\end{lemma}

\begin{proof}
 Without loss of generality assume $z_f(C)_1 \cap z_f(D)_1 \neq \emptyset$.
 Let $p^{C} \in z_f(C)$ and $p^{D} \in z_f(D)$ such that $p^{C}_1 = p^{D}_1$.
 Again without loss of generality assume  $p^{C}_2 < p^{D}_2$.
 Let $v = \text{argmin}_{v \in C} f(v)_2$ and $w \in D \setminus C$.
 Then $\|p^{C} - f(w)\|_{\infty} > \frac{1}{2}$ and $\|p^{D} - f(w)\|_{\infty} \leq \frac{1}{2}$.
 Suppose $vw \in E(G)$. Then $0 < f(w)_2 - f(v)_2 \leq 1$.
 Let $p = (p^{C}_1, f(v)_2 + \frac{1}{2})$.
 Then $\|p - f(w)\|_{\infty} \leq \frac{1}{2}$, since $p -f(w) = (p^{D}_1 - f(w)_1,f(v)_2 - f(w)_2 + \frac{1}{2})$.
 Further $p \in z_f(C)$ by Lemma \ref{la:clique-center-border}.
 Thus, $C \cup \{w\}$ forms a clique contradicting the maximality of $C$.
 So there is no edge between $v$ and any $w \in D \setminus C$.
 But then $v \not\simeq w$, because $\simeq$ describes a clique-partition, and thus $C \not\simeq D$.
\end{proof}

The last lemma gives us the possibility to order equivalent cliques with respect to some realization $f$.
Each $z_f(C)_i$ forms an interval according to Lemma \ref{la:clique-center-border}.
So let $C \simeq D$ be two maximal cliques and define $z_f(C)_i \leq z_f(D)_i$ if for some (and thus for all) $p^{C}_i \in z_f(C)_i$, $p^{D}_i \in z_f(C)_i$ it holds $p^{C}_i \leq p^{D}_i$.
For each equivalence class of maximal cliques this defines two linear orders on the maximal cliques.
A central observation is that the two linear orders either coincide or one is the reverse order of the other.

\begin{lemma}
 \label{la:order-clique-classes}
 For each $\mathcal{C} \in \mathcal{M}(G)/_{\simeq}$ there is some $b \in \{-1,1\}$, such that for every $C,D \in \mathcal{C}$:
 \begin{equation}
  z_f(C)_1 \leq z_f(D)_1 \Leftrightarrow b \cdot z_f(C)_2 \leq b \cdot z_f(D)_2.
 \end{equation}
 The value $b$ is called the \emph{orientation} of $\mathcal{C}$ with respect to $f$.
\end{lemma}

\begin{figure}
 \centering
 \begin{tikzpicture}
  \draw[step=0.2, gray!50, very thin] (0.1,0.1) grid (5.3,4.9);
  
  \draw (0.8,1.0) -- (3.8,1.0) -- (3.8,4.0) -- (0.8,4.0) -- (0.8,1.0);
  \draw (1.2,1.4) -- (4.2,1.4) -- (4.2,4.4) -- (1.2,4.4) -- (1.2,1.4);
  \draw (1.6,0.6) -- (4.6,0.6) -- (4.6,3.6) -- (1.6,3.6) -- (1.6,0.6);
  
  \draw[fill] (0.9,2.5) circle (0.03) node[above] {$v$};
  \draw[dashed, line width=1.5pt, color=gray] (3.8,0.4) -- (3.8,4.6);
  
  \draw[fill] (2.7,4.2) circle (0.03) node[right] {$w$};
  \draw[dashed, line width=1.5pt, color=gray] (0.4,1.0) -- (5.0,1.0);
  
 \end{tikzpicture}
 \caption{Each equivalence class of maximal cliques is diagonally ordered}
 \label{fig:unit-square-clique-order}
\end{figure}
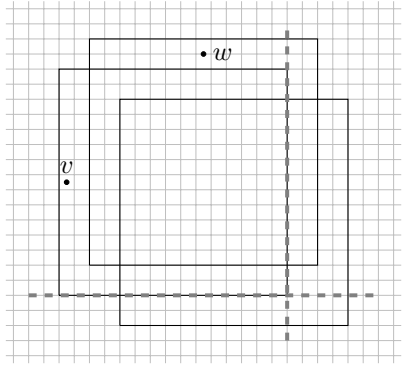

\begin{proof}
 Suppose not and let $\mathcal{C} \in \mathcal{M}(G)/_{\simeq}$ be an equivalence class violating the condition.
 Let $<$ be the linear order on $\mathcal{C}$ induced by ordering the intervals $z_f(C)_1$ for $C \in \mathcal{C}$ in the natural way from left to right.
 Let $C_1 < C_2 < \dots < C_m$ be the elements of $\mathcal{C}$ and define $b \in \{-1,1\}$, such that
 \begin{equation*}
  b \cdot z_f(C_1)_2 \leq b \cdot z_f(C_2)_2.
 \end{equation*}
 Without loss of generality assume $b = 1$.
 Let $i \in [m]$ be minimal so that $z_f(C_{i+1})_2 \nleq z_f(C_{i+2})_2$.
 Consider the case where $z_f(C_{i})_2 \geq z_f(C_{i+2})_2$ (cf.\ Figure \ref{fig:unit-square-clique-order}).
 The other case is $z_f(C_{i})_2 \leq z_f(C_{i+2})_2 \leq z_f(C_{i+1})_2$ and is analogous by a rotation of $90$ degree.
 
 Let $v = \text{argmin}_{v \in C_i} f(v)_1$.
 Since $C_i \simeq C_{i+1}$ and $v \notin C_{i+1}$ there is some $w \in C_{i+1} \setminus C_i$ with $v \simeq w$.
 In particular $vw \in E(G)$.
 Using the same argument there is some $u \in C_{i+2} \setminus C_i$ with $u \simeq v$ and thus, $uv,uw \in E(G)$.
 Additionally considering the positions of the maximal cliques we get $f(v)_1 \leq f(u)_1 \leq f(v)_1 + 1$.
 \vspace*{5pt}\\
 \textit{Claim 1:} There is some $p \in z_f(C_i)$, such that $|p_1 - f(u)_1| \leq \frac{1}{2}$.
 \vspace*{5pt}\\
 Choose $p = (f(v)_1 + \frac{1}{2}, \text{min}_{v' \in C_i} f(v')_2 + \frac{1}{2})$.
 Then $p \in z_f(C_i)$ by Lemma \ref{la:clique-center-border}.
 Further $|p_1 - f(u)_1| = |f(v)_1 - f(u)_1 + \frac{1}{2}| \leq \frac{1}{2}$.
 \vspace*{5pt}\\
 \textit{Claim 2:} For every $p \in z_f(C_i)$ it holds that $|p_2 - f(u)_2| \leq \frac{1}{2}$.
 \vspace*{5pt}\\
 Let $p \in z_f(C_i)$.
 Then $p_2 + \frac{1}{2} \geq f(u)_2$.
 First show $p_2 + \frac{1}{2} \leq f(w)_2$.
 Let $p' = (f(v)_1 + \frac{1}{2}, \text{min}_{v' \in C_i} f(v')_2 + \frac{1}{2}$.
 It is $f(v)_1 \leq f(w)_1 \leq f(v)_1 + 1$ and thus, $|p'_1 - f(w)| \leq \frac{1}{2}$.
 Since $w \notin C_i$ we get $|p'_2 - f(w)| > \frac{1}{2}$ and the considering positions of the maximal cliques even $f(w)_2 - p'_2 > \frac{1}{2}$.
 Hence $p_2 + \frac{1}{2} \geq f(u)_2$ by Lemma \ref{la:clique-center-border}.
 As a consequence we also get $f(w)_2 \geq f(u)_2 \geq f(w)_2 - 1$.
 So all together $-\frac{1}{2} \leq f(u)_2 -\frac{1}{2} - f(u)_2 \leq p_2 - f(u)_2 \leq p_2 - f(w)_2 + 1 \leq f(w)_2 - \frac{1}{2} - f(w)_2 + 1 = \frac{1}{2}$.
 \vspace*{5pt}\\
 So there is some $p \in z_f(C_i)$ with $\|p - f(u)\|_{\infty} \leq \frac{1}{2}$.
 But this is a contradiction to $u \notin C_i$.
 A visualization is also given in Figure \ref{fig:unit-square-clique-order}.
\end{proof}

\begin{corollary}
 \label{cor:order-clique-classes}
 Let $\mathcal{C} \in \mathcal{M}(G)/_{\simeq}$ and for $C,D \in \mathcal{C}$ define $C <_{f,\mathcal{C}} D$ if $z_f(C)_1 \leq z_f(D)_1$.
 Then $\mathcal{M}(v) \cap \mathcal{C}$ forms an interval with respect to $<_{f,\mathcal{C}}$ for each $v \in V(G)$.
\end{corollary}

\begin{proof}
 Let $v \in V(G)$, such that there are distinct $C,D \in \mathcal{M}(v) \cap \mathcal{C}$.
 Without loss of generality let $C <_{f,\mathcal{C}} D$.
 Let $E \in \mathcal{C}$ with $C <_{f,\mathcal{C}} E <_{f,\mathcal{C}} D$.
 Pick points $p^{A} \in z_f(A)$ for $A \in \{C,D,E\}$.
 Then $p^{C}_1 < p^{E}_1 < p^{D}_1$.
 If $f(v)_1 \geq p^{E}_1$ then $|f(v)_1 - p^{E}_1| < |f(v)_1 - p^{C}_1| \leq \frac{1}{2}$.
 Otherwise $|f(v)_1 - p^{E}_1| < |f(v)_1 - p^{D}_1| \leq \frac{1}{2}$.
 
 Now let $b \in \{-1,1\}$ be the orientation of $\mathcal{C}$.
 Consider $b = 1$. The other case is analogous.
 Then $p^{C}_2 < p^{E}_2 < p^{D}_2$ and using the same argument as before $|f(v)_2 - p^{E}_2| \leq \frac{1}{2}$.
 So $v \in E$ and thus, $\mathcal{M}(v) \cap \mathcal{C}$ forms an interval with respect to $<_{f,\mathcal{C}}$.
\end{proof}

\begin{lemma}
 \label{la:equivalent-vertex-direction}
 Let $C \simeq D$ be two maximal cliques with $z_f(C)_i \leq z_f(D)_i$ for $i \in [2]$.
 Further let $v \in C \setminus D$ and $w \in D \setminus C$ with $v \simeq w$.
 Then $f(v)_i < f(w)_i$ for $i \in [2]$.
\end{lemma}

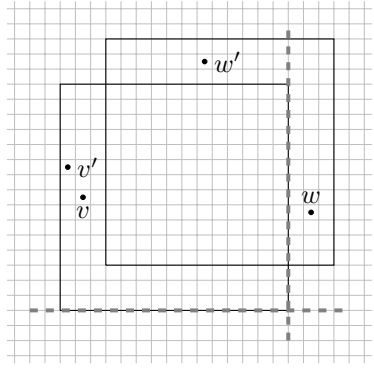
\begin{figure}
 \centering
 \begin{tikzpicture}
  \draw[step=0.2, gray!50, very thin] (0.1,0.1) grid (4.9,4.9);
  
  \draw (0.8,0.8) -- (3.8,0.8) -- (3.8,3.8) -- (0.8,3.8) -- (0.8,0.8);
  \draw (1.4,1.4) -- (4.4,1.4) -- (4.4,4.4) -- (1.4,4.4) -- (1.4,1.4);
  
  \draw[fill] (1.1,2.3) circle (0.03) node[below] {$v$};
  \draw[fill] (0.9,2.7) circle (0.03) node[right] {$v'$};
  \draw[dashed, line width=1.5pt, color=gray] (0.4,0.8) -- (4.6,0.8);
  
  \draw[fill] (4.1,2.1) circle (0.03) node[above] {$w$};
  \draw[fill] (2.7,4.1) circle (0.03) node[right] {$w'$};
  \draw[dashed, line width=1.5pt, color=gray] (3.8,0.4) -- (3.8,4.6);

 \end{tikzpicture}
 \caption{Extending the orientation on equivalence classes of cliques to equivalent vertices}
 \label{fig:equivalent-vertex-direction}
\end{figure}

\begin{proof}
 Let $p^{D} \in z_f(D)$.
 Further choose $i \in [2]$ in such a way that $|f(v)_i - p^{D}_i| > \frac{1}{2}$.
 Without loss of generality assume $i = 1$.
 Then $f(w)_1 > f(v)_1$ because $f(v)_1 < p^{D}_1 - \frac{1}{2}$.
 Let $v' = \text{argmin}_{v \in C}f(v)_1$ and $p^{C} = (f(v') + \frac{1}{2}, \text{min}_{v \in C}f(v)_2 + \frac{1}{2})$.
 Then $p^{C} \in z_f(C)$ by Lemma \ref{la:clique-center-border} and $p^{C}_1 < p^{D}_1$ which implies $v' \notin D$.
 Let $w' \in D \setminus C$ with $v' \simeq w'$.
 Then $f(v')_1 < f(w')_1 \leq f(v')_1 + 1$ and thus, $|p^{C}_1 - f(w')_1| \leq \frac{1}{2}$.
 Hence, $f(w') > p^{C}_2 + \frac{1}{2}$.
 Let $u' \simeq v'$.
 Then $u'v' \in E(G)$ and $u'w' \in E(G)$.
 So $f(u')_1 \leq p^{C}_1 + \frac{1}{2}$ and $f(u')_2 \geq p^{C}_2 - \frac{1}{2}$.
 Now suppose towards a contradiction that $f(w)_2 \leq f(v)_2$.
 Then $|f(w)_2 - p^{C}_2| \leq \frac{1}{2}$ because $p^{C}_2 < p^{D}_2$.
 \vspace*{5pt}\\
 \textit{Claim:} If $wu' \in E(G)$ then $vu' \in E(G)$.
 \vspace*{5pt}\\
 Let $wu' \in E(G)$.
 If $f(u')_1 \geq f(v)_1$ then $0 \leq f(u')_1 - f(v)_1 \leq p^{C}_1 + \frac{1}{2} - f(v)_1 \leq 1$.
 Otherwise $f(u')_1 < f(v)_1$ and $0 < f(v)_1 - f(u')_1 < f(w)_1 - f(u')_1 \leq 1$.
 So $|f(v)_1 - f(u')_1| \leq 1$.
 Similarly $|f(v)_2 - f(u')_2| \leq 1$ and thus, $vu' \in E(G)$.
 \vspace*{5pt}\\
 But $wv' \notin E(G)$ because $f(w)_1 - f(v') > 1$.
 Otherwise $f(w) - p^{C}_1 \leq \frac{1}{2}$ and $w \in C$.
 This is a contradiction to $\simeq$ being a stable partition.
 A visualization is also given in Figure \ref{fig:equivalent-vertex-direction}.
\end{proof}

\begin{corollary}
 \label{cor:equivalent-vertex-direction}
 Let $v,w \in V(G)$ with $v \simeq w$ and $\mathcal{C} \in \mathcal{M}(G)/_{\simeq}$, such that $\mathcal{M}(v) \cap \mathcal{C} \neq \mathcal{M}(w) \cap \mathcal{C}$.
 Further let $b \in \{-1,1\}$ be the orientation $\mathcal{C}$.
 Then $f(v)_i \neq f(w)_i$ for $i \in [2]$ and 
 \begin{equation}
  f(v)_1 < f(w)_1 \;\Leftrightarrow\; b \cdot f(v)_2 < b \cdot f(w)_2.
 \end{equation}
\end{corollary}

\begin{proof}
 Without loss of generality assume $b = 1$.
 Let $\mathcal{C}(u) = \mathcal{M}(u) \cap \mathcal{C}$ for $u \in V(G)$.
 Since $|\mathcal{C}(v)| = |\mathcal{C}(w)|$ there are $C \in \mathcal{C}(v)\setminus\mathcal{C}(w)$ and $D \in \mathcal{C}(w)\setminus\mathcal{C}(v)$.
 Then the statement follows from Lemma \ref{la:equivalent-vertex-direction}.
\end{proof}

\begin{corollary}
 \label{cor:unit-square-weakly-ordered}
 Let $\mathcal{C}, \mathcal{D} \in \mathcal{M}(G)/_{\simeq}$.
 Further let $C,C' \in \mathcal{C}$ with $C <_{f,\mathcal{C}} C'$ and $D,D' \in \mathcal{D}$,
 such that there are vertices $v,w \in V(G)$ with $v \simeq w$ and $v \in (C \setminus C') \cap (D \setminus D')$, $w \in (C' \setminus C) \cap (D' \setminus D)$.
 Then $D <_{f,\mathcal{D}} D'$.
\end{corollary}

\begin{proof}
 Let $b_{\mathcal{C}},b_{\mathcal{D}} \in \{-1,1\}$ be the orientations of $\mathcal{C}$ and $\mathcal{D}$ respectively.
 Without loss of generality assume $b_{\mathcal{C}} = 1$.
 Then $b_{\mathcal{D}} = 1$ by Corollary \ref{cor:equivalent-vertex-direction}.
 So $D <_{f,\mathcal{D}} D'$ by Lemma \ref{la:equivalent-vertex-direction}.
\end{proof}

We are now ready to prove the main structural lemma on clique-stable partitions.

\begin{lemma}
 \label{la:internally-stable-partition}
 Let $G$ be a twin-free unit square graph and let $\mathcal{P}$ be a clique-stable partition.
 Further let $f\colon V(G) \rightarrow \mathbb{R}^{2}$ be a realization.
 Then the following properties hold:
 \begin{enumerate}
  \item For each $X \in \mathcal{P}$ there exists $b \in \{-1,1\}$, such that
   \begin{equation}
    \label{eq:vertex-diagonal-order}
    f(v)_1 \leq f(w)_1 \;\Leftrightarrow\; b \cdot f(v)_2 \leq b \cdot f(w)_2
   \end{equation}
   for all $v,w \in X$. The value $b$ is called the \emph{orientation of $X$}, which is denoted by $\ori_{G,f}(X)$
   (the value $b$ is unique unless $|X| = 1$, in this case we define $\ori_{G,f}(X) = 1$).
  \item Let $X,Y \in \mathcal{P}$ with $\ori_{G,f}(X) \neq \ori_{G,f}(Y)$.
   Then either $xy \in E(G)$ for all $x \in X,y \in Y$ or there is no $x \in X,y \in Y$ with $xy \in E(G)$.
  \item Let $X,Y \in \mathcal{P}$ with $\ori_{G,f}(X) = \ori_{G,f}(Y)$ and suppose $k = |\{xy \in X \times Y \mid xy \in E(G)\}| \geq 1$.
   Further let $X = \{x_1,\dots,x_s\}$, such that $f(x_i)_1 \leq f(x_{i+1})_1$ for all $i \in [s]$, and $Y = \{y_1,\dots,y_t\}$, such that $f(y_j)_1 \leq f(y_{j+1})_1$ for all $j \in [t]$.
   Then
   \begin{equation}
    \label{eq:diagonal-edges}
    x_iy_j \in E(G) \;\;\Leftrightarrow\;\; \left\lceil \frac{i \cdot t}{k} \right\rceil = \left\lceil \frac{j \cdot s}{k} \right\rceil.
   \end{equation}
 \end{enumerate}
\end{lemma}

\begin{proof}
 Let $X \in \mathcal{P}$ and let $X = \{v_1,\dots,v_k\}$ with $f(v_i)_1 \leq f(v_j)_1$ for all $i \leq j \in [k]$.
 Without loss of generality assume $f(v_1)_2 \leq f(v_2)_2$.
 In this case $b = 1$.
 Suppose towards a contradiction that there is some $i \in [k]$ with $f(v_{i+1})_2 > f(v_{i+2})_2$.
 Choose $i \in [k]$ minimal with this property.
 Without loss of generality assume that $f(v_i)_2 \leq f(v_{i+2})_2$.
 The other case is analogous.
 Let $\mathcal{P}^{*}$ be a stable partition of the graph $G_{\mathcal{M}}^{*}$, such that $\mathcal{P}^{*} \cap 2^{V(G)} = \mathcal{P}$.
 Let $\mathcal{C} \in \mathcal{P}_{N} \cap 2^{\mathcal{M}(G)}$, such that $\mathcal{M}(v_{i+1}) \cap \mathcal{C} \neq \mathcal{M}(v_{i+2}) \cap \mathcal{C}$.
 Since $|\mathcal{M}(v_{i+1}) \cap \mathcal{C}| = |\mathcal{M}(v_{i+2}) \cap \mathcal{C}|$ there are $C,C' \in \mathcal{C}$ with $v_{i+1} \in C \setminus C'$ and $v_{i+2} \in C' \setminus C$.
 Let $b$ be the orientation of $\mathcal{C}$.
 Then $b = -1$ by Corollary \ref{cor:equivalent-vertex-direction}.
 Furthermore $\mathcal{M}(v_{i}) \cap \mathcal{C} = \mathcal{M}(v_{i+1}) \cap \mathcal{C}$ and $\mathcal{M}(v_{i}) \cap \mathcal{C} = \mathcal{M}(v_{i+2}) \cap \mathcal{C}$ by Corollary \ref{cor:equivalent-vertex-direction}.
 But this is a contradiction.
 
 So let $X,Y \in \mathcal{P}$ with $\ori_{G,f}(X) \neq \ori_{G,f}(Y)$.
 Suppose there is some edge between $X$ and $Y$.
 Without loss of generality assume $\ori_{G,f}(X) = 1$.
 Consider $x = \argmax_{x \in X} f(x)_1$ and $y = \argmax_{y \in Y} f(y)_1$.
 \vspace*{10pt}\\
 \textit{Case 1:} $f(x)_1 \leq f(y)_1$ and $f(x)_2 \leq f(y)_2$.
 \vspace*{5pt}\\
 Let $y' \in Y$.
 First, there is some $x ' \in X$ with $x'y \in E(G)$, because the partition $\mathcal{P}$ is stable.
 Then $f(y)_1 - 1 \leq f(x')_1 \leq f(x)_1 \leq f(y)_1$.
 Additionally $f(y)_1 - 1 \leq f(y')_1 \leq f(y)_1$.
 Hence, $|f(x)_1 - f(y')_1| \leq 1$.
 Furthermore, by stability, there is some $x'' \in X$ with $x''y' \in E(G)$.
 Then $f(y')_2-1 \leq f(x'')_2 \leq f(x)_2 \leq f(y)_2 \leq f(y')_2$ and thus, $|f(x)_2 - f(y')_2| \leq 1$.
 So $xy' \in E(G)$ for all $y' \in Y$.
 By stability it follows that $x'y' \in E(G)$ for all $x' \in X,y' \in Y$.
 \vspace*{10pt}\\
 \textit{Case 2:} $f(x)_1 \geq f(y)_1$ and $f(x)_2 \leq f(y)_2$.
 \vspace*{5pt}\\
 The proof is analogous to the first case by switching the roles of $X$ and $Y$.
 \vspace*{10pt}\\
 \textit{Case 3:} $f(x)_1 \leq f(y)_1$ and $f(x)_2 \geq f(y)_2$.
 \vspace*{5pt}\\
 Let $y^{*} = \argmin_{y \in Y}f(y)_1$ and first suppose $f(x)_2 \geq f(y^{*})_2$.
 Let $x' \in X$.
 Then, by stability, there is some $y' \in Y$ with $x'y' \in E(G)$.
 If $f(x')_1 \leq f(y')_1$ then $f(y')_1 - 1 \leq f(x')_1 \leq f(y')_1$ and $f(y')_1 - 1 \leq f(y^{*})_1 \leq f(y')_1$.
 Otherwise $f(y^{*})_1 \leq f(y')_1 \leq f(x')_1 \leq f(x)_1 \leq f(y)_1 \leq f(y^{*})_1 + 1$.
 So in both cases it holds that $|f(x')_1 - f(y^{*})_1| \leq 1$.
 Furthermore there is some $y''$ with $xy'' \in E(G)$.
 So $f(x)_2 \geq f(y^{*})_2 \geq f(y'')_2 \geq f(x)_2 + 1$.
 Additionally $f(x)_2 \geq f(x')_2 \geq f(x)_2 + 1$ and thus, $|f(x')_2 - f(y^{*})_2| \leq 1$.
 So $y^{*}x' \in E(G)$ for all $x' \in X$.
 By stability it follows again that $x'y' \in E(G)$ for all $x' \in X,y' \in Y$.
 
 So consider $f(x)_2 \leq f(y^{*})_2$.
 Then $f(y)_2 +1 \geq f(y^{*})_2 \geq f(x)_2 \geq f(y)_2$.
 Let $y' \in Y$.
 Then $f(y)_2 +1 \geq f(y')_2 \geq f(y)_2$ and consequently, $|f(x)_2 -f(y')_2| \leq 1$.
 Furthermore, there is some $x ' \in X$ with $x'y \in E(G)$.
 Then $f(y)_1 - 1 \leq f(x')_1 \leq f(x)_1 \leq f(y)_1$.
 Additionally $f(y)_1 - 1 \leq f(y')_1 \leq f(y)_1$.
 Hence, $|f(x)_1 - f(y')_1| \leq 1$.
 So $xy' \in E(G)$ for all $y' \in Y$.
 By stability it follows once more that $x'y' \in E(G)$ for all $x' \in X,y' \in Y$.
 \vspace*{10pt}\\
 \textit{Case 4:} $f(x)_1 \geq f(y)_1$ and $f(x)_2 \geq f(y)_2$.
 \vspace*{5pt}\\
 The proof is analogous to the third case by switching the roles of $X$ and $Y$.
 \vspace*{10pt}
 
 It remains to prove the third statement.
 Let $X,Y \in \mathcal{P}$ with $\ori_{G,f}(X) = \ori_{G,f}(Y)$.
 Without loss of generality assume that $\ori_{G,f}(X) = 1$.
 First, $|N(x) \cap Y| = \frac{k}{s}$ and $|N(y) \cap X| = \frac{k}{t}$ for each $x \in X,y \in Y$, because $\mathcal{P}$ is clique-stable.
 For a vertex $x \in X$ let $I_x = \{j \in [t] \mid xy_j \in E(G)\}$.
 \vspace*{10pt}\\
 \textit{Claim 1:} For each $x \in X$ the set $I_x$ is an interval.
 \vspace*{5pt}\\
 Let $j_1 < j_2 < j_3 \in [t]$ with $j_1,j_3 \in I_x$.
 If $f(y_{j_2})_1 \leq f(x)_1$ then $f(y_{j_1})_1 \leq f(y_{j_2})_1 \leq f(x)_1 \leq f(y_{j_1})_1 + 1$. 
 Otherwise $f(y_{j_3})_1 - 1 \leq f(x)_1 \leq f(y_{j_2})_1 \leq f(y_{j_3})_1$.
 In both cases $|f(x) - f(y_{j_2})_1| \leq 1$.
 The same argument also works for the second dimension by Lemma \ref{la:internally-stable-partition}.
 So $j_2 \in I_x$.
 \vspace*{10pt}\\
 For $y \in Y$ the set $I_y$ is defined analogously, so $I_y = \{i \in [s] \mid x_iy \in E(G)\}$.
 Clearly the previous claim also holds for all $y \in Y$.
 \vspace*{10pt}\\
 \textit{Claim 2:} Let $i < i' \in [s]$ and $j < j' \in [t]$, such that there are edges $x_iy_{j'}, x_{i'}y_j \in E(G)$.
  Then $x_iy_j, x_{i'}y_{j'} \in E(G)$.
 \vspace*{5pt}\\
 Let $r \in [2]$.
 If $f(x_i)_r \leq f(y_j)_r$ then $f(y_{j'})_r - 1 \leq f(x_i)_r \leq f(y_j)_r \leq f(y_{j'})_r$.
 Otherwise $f(x_{i'})_r - 1 \leq f(y_j)_r \leq f(x_i)_r \leq f(x_{i'})_r$.
 In both cases $|f(x_i)_r - f(y_j)_r| \leq 1$.
 Hence, $x_iy_j \in E(G)$.
 Analogously it holds that $x_{i'}y_{j'} \in E(G)$.
 \vspace*{10pt}\\
 For $r \in [\frac{s \cdot t}{k}]$ let $X^{r} = \{x_i \mid \lceil \frac{i \cdot t}{k} \rceil = r\}$ and $Y^{r} = \{y_j \mid \lceil \frac{j \cdot s}{k} \rceil = r\}$.
 Then $|X^{r}| = \frac{k}{t}$ and $|Y^{r}| = \frac{k}{s}$.
 It suffices to show that for all $r \in [\frac{s \cdot t}{k}]$ it holds that $xy \in E(G)$ for all $x \in X^{r}, y \in Y^{r}$.
 Let $r \in [\frac{s \cdot t}{k}]$ and suppose the statement holds for all $r' < r$.
 Let $i \in [s]$ be minimal with $x_i \in X^{r}$ and let $j \in [t]$ be minimal with $y_j \in Y^{r}$.
 We first show that $x_iy_j \in E(G)$
 Let $i' \in [s], j' \in [t]$ be such that $y_jx_{i'} \in E(G)$ and $x_iy_{j'} \in E(G)$.
 But then $i \leq i'$ since for every $i'' < i$ the vertex $x_{i''}$ already has the correct number of neighbors.
 Analogously it holds that $j \leq j'$.
 So $x_iy_j \in E(G)$ by Claim 2.
\end{proof}

Now let $\mathcal{P}$ be a canonical, clique-stable partition of the graph $G$.
We define the vertex and edge-colored graph $G[\mathcal{P}]=(\mathcal{P},E,c_V,c_E)$ with $E = \{XY \mid X \neq Y \in \mathcal{P}\}$,
\begin{equation*}
 c_V\colon\mathcal{P} \rightarrow \mathbb{N}\colon X \mapsto |X|
\end{equation*}
and
\begin{equation*}
 c_E\colon E \rightarrow \mathbb{N}\colon XY \mapsto |\{xy \in X \times Y \mid xy \in E(G)\}|.
\end{equation*}
For $\gamma \in \Aut(G)$ define the permutation $\gamma^{\mathcal{P}} := \varphi(\gamma) \in \Sym(\mathcal{P})$ where $\varphi\colon\Aut(G) \rightarrow \Sym(\mathcal{P})$ is the natural action of $\Aut(G)$ on $\mathcal{P}$.
Note that $\varphi$ is well-defined since the partition $\mathcal{P}$ is canonical.

\begin{theorem}
 \label{thm:global-to-local}
 Let $G$ be a twin-free unit square graph and let $\mathcal{P}$ be a canonical, clique-stable partition.
 Further let $\delta \in \Aut(G[\mathcal{P}])$.
 Then there is some $\gamma \in \Aut(G)$ with $\gamma^{\mathcal{P}} = \delta$.
\end{theorem}

\begin{proof}
 Let $f\colon V(G) \rightarrow \mathbb{R}^{2}$ be a realization of $G$.
 For $X \in \mathcal{P}$ let $X = \{v_{X,1},\dots,v_{X,|X|}\}$, such that $f(v_{X,i})_1 \leq f(v_{X,j})_1$ for all $i \leq j \in [|X|]$.
 Define
 \begin{equation*}
  \gamma\colon V(G) \rightarrow V(G)\colon v_{X,i} \mapsto v_{X^{\delta},i}.
 \end{equation*}
 First observe that $\gamma$ is well-defined since $\mathcal{P}$ forms a partition and $|X| = |X^{\delta}|$ for all $X \in \mathcal{P}$.
 Furthermore $\gamma$ is bijective because $\delta$ is bijective.
 Thus, $\gamma \in \Sym(V(G))$ and clearly $\gamma^{\mathcal{P}} = \delta$.
 So it remains to prove that $\gamma$ preserves the edge relation.
 
 Let $vw \in E(G)$ and let $X,Y \in \mathcal{P}$, such that $v \in X$ and $w \in Y$.
 If $X = Y$ then $v^{\gamma},w^{\gamma} \in X^{\delta}$ and thus, $v^{\gamma}w^{\gamma} \in E(G)$ because $X$ is a clique.
 So suppose $X \neq Y$.
 First consider the case that $\ori_{G,f}(X) \neq \ori_{G,f}(Y)$.
 Then $c_E(XY) = |X|\cdot|Y|$ by Lemma \ref{la:internally-stable-partition}.
 Hence, $c_E(X^{\delta}Y^{\delta}) = |X|\cdot|Y| = |X^{\delta}|\cdot|Y^{\delta}|$ and $xy \in E(G)$ for every $x \in X^{\delta},y \in Y^{\delta}$.
 In particular $v^{\gamma}w^{\gamma} \in E(G)$.
 
 So suppose $\ori_{G,f}(X) = \ori_{G,f}(Y)$.
 Without loss of generality assume that $\ori_{G,f}(X) = 1$.
 Let $k = c_E(XY)$, $s = |X|$ and $t = |Y|$.
 Then $k \geq 1$ and suppose $k < |X|\cdot|Y|$ (otherwise the statement follows with same argument as before).
 Further let $x_i = v_{X,i}$ for $i \in [s]$ and $y_j = v_{Y,j}$ for $j \in [t]$.
 Now pick $i \in [s], j \in [t]$, such that $v = x_i = v_{X,i}$ and $w = y_j = v_{Y,j}$.
 Then $\lceil \frac{i \cdot s_Y}{k} \rceil = \lceil \frac{j \cdot s_X}{k} \rceil$ by Lemma \ref{la:internally-stable-partition}.
 Further $s = |X^{\delta}|$, $t = |Y^{\delta}|$ and $k = c_E(X^{\delta}Y^{\delta}) < |X^{\delta}| \cdot |Y^{\delta}|$.
 Hence, $\ori_{G,f}(X^{\delta}) = \ori_{G,f}(Y^{\delta})$ by Lemma \ref{la:internally-stable-partition}.
 So $v^{\gamma}w^{\gamma} = v_{X^{\delta},i}v_{Y^{\delta},j} \in E(G)$ again by Lemma \ref{la:internally-stable-partition}.
\end{proof}

Intuitively, the last theorem states that each automorphism of $G[\mathcal{P}]$ naturally extends to an automorphism of $G$.
In particular, the graph $G$ can be uniquely reconstructed from the graph $G[\mathcal{P}]$.
This is the main result on the local structure of unit square graphs which allows us, for a canonical, clique-stable partition $\mathcal{P}$, to restrict to the graph $G[\mathcal{P}]$.
For the remainder of this work the goal is to compute a canonical, clique-stable partition $\mathcal{P}$, such that the automorphism group of $G[\mathcal{P}]$ is a $\Gamma_t$-group for some constant $t$.
To achieve this goal we require the graph $G$ to have some singleton vertex $v_0$ (a vertex with a unique color).
More precisely, for such a graph we construct the desired partition $\mathcal{P}$ and a canonical, $t$-circle-bounded graph $H$, such that $\mathcal{P} \subseteq V(H)$ and $\mathcal{P}$ is invariant under $\Aut(H)$.
This results in a good supergroup of the group $\Aut(G[\mathcal{P}])$ which can be used by Luks' algorithm to compute the real automorphism group.
To compute the graph $H$ we devise an algorithm that iteratively extends $H$ taking vertices with larger and larger distances to $v_0$ into account.
While doing so the crucial subproblem is to compute canonical clique-partitions for neighborhoods of cliques.
This problem is addressed in the next section.

\section{Neighborhoods}
\label{sec:neighborhoods}

In order to obtain canonical clique-partitions for neighborhoods we essentially proceed in two steps.
First, we use some combinatorial partitioning techniques to obtain some initial coloring of the vertices.
Then, considering each color class separately, the main contribution is to prove that each color class can either be described by a co-bipartite graph or a proper circular arc graph.
In both cases it is easy to compute a canonical clique-partition.

\subsection{Two structure Lemmas}

We first require two general graph-theoretic lemmas.

\begin{lemma}
 \label{la:forbid-complement-c6}
 Let $G$ be a graph, such that
 \begin{enumerate}
  \item \label{item:forbid-complement-c6-1} $G[N[v]]$ is a unit interval graph for each $v \in V(G)$,
  \item \label{item:forbid-complement-c6-2} $G$ has no induced subgraph isomorphic to $C_4 \cup K_1$.
 \end{enumerate}
 Further let $X = \{w_1 \in V(G) \mid \exists w_2,\dots,w_6\colon G[w_1,\dots,w_6] \cong \overline{C_6}\}$.
 Then $G[X]$ is co-bipartite.
\end{lemma}

\begin{proof}
 Suppose $X \neq \emptyset$ and let $W = \{w_{i,j} \mid i \in [2],j \in [3]\} \subseteq X$, such that $G[W] \cong \overline{C_6}$ as depicted in Figure \ref{fig:graph-c6}.
 If $X = W$ then $G[X] \cong \overline{C_6}$ which is co-bipartite.
 So assume $W \subsetneq X$ and let $v \in X \setminus W$.
 Then there are $i_1,i_2 \in [2]$ and $j_1,j_2 \in [3]$ with $j_1 \neq j_2$, such that $vw_{i_s,j_s} \in E(G)$ for $s \in [2]$, because of property \ref{item:forbid-complement-c6-2}.
 \vspace*{10pt}\\
 \textit{Claim:} There is some $i \in [2]$, such that $vw_{i,j} \in E(G)$ for all $j \in [3]$.
 \vspace*{5pt}\\
 Suppose this does not hold and consider some vertex $w_{i,j}$ with $vw_{i,j} \in E(G)$.
 Then there is some $j' \in [3]$, such that $vw_{i,j'} \notin E(G)$.
 So $vw_{3-i,j} \in E(G)$, because otherwise the graph $G[v,w_{i,j},w_{i,j'},w_{3-i,j}]$ would be isomorphic to a claw contradicting property \ref{item:forbid-complement-c6-1}.
 Hence, $vw_{i,j} \in E(G)$ if and only if $vw_{3-i,j} \in E(G)$ for each $j \in [3]$.
 But then $G[N[v]]$ contains a cycle of length $4$, since there are two distinct $j_1,j_2 \in [3]$ with $vw_{i,j_1},vw_{i,j_2} \in E(G)$.
 This is a contradiction to property \ref{item:forbid-complement-c6-1}.
 \vspace*{10pt}\\
 Now let $v' \in X \setminus W$ be another vertex, such that $v'w_{i,j} \in E(G)$ for all $j \in [3]$.
 Suppose that $vv' \notin E(G)$ and let $j \in [3]$.
 Then $uw_{3-i,j} \in E(G)$ for some $u \in \{v,v'\}$, since otherwise the graph $G[w_{i,j},v,v',w_{3-i,j}]$ would be isomorphic to a claw.
 So for some $u \in \{v,v'\}$ there are distinct $j_1,j_2 \in [3]$ with $uw_{3-i,j_1},uw_{3-i,j_2} \in E(G)$.
 But this again contradicts the fact, that $G[N[u]]$ does not contain a cycle of length $4$.
 
 Now let $X_1 = \{v \in X \mid \forall j \in [3]\colon vw_{1,j} \in E(G)\}$ and $X_2 = X \setminus X_1$.
 Then $X_1$ is a clique by the comment below the Claim.
 For every $v \in X_2$ it holds that $vw_{2,j} \in E(G)$ for all $j \in [3]$ by the Claim.
 So $X_2$ is also a clique using again the comment below the Claim.
\end{proof}

\begin{lemma}
 \label{la:forbid-net}
 Let $G$ be a graph, such that
 \begin{enumerate}
  \item \label{item:forbid-net-1} $G[N[v]]$ is a unit interval graph for each $v \in V(G)$,
  \item \label{item:forbid-net-2} $G$ has no induced subgraph isomorphic to $C_{n+4} \cup K_1$ for $n \geq 0$,
  \item \label{item:forbid-net-3} there are no $v,w \in V(G)$, such that $N[v] \subsetneq N[w]$.
 \end{enumerate}
 Then $G$ has no induced subgraph isomorphic to \texttt{net}.
\end{lemma}

\begin{proof}
 Assume towards a contradiction that $G$ has an induced subgraph isomorphic to \texttt{net}.
 So there are vertices $v_1,v_2,v_3,w_1,w_2,w_3 \in V(G)$, such that $G[v_1,v_2,v_3,w_1,w_2,w_3] \cong\texttt{net}$ as depicted in Figure \ref{fig:graph-net}.
 Further let $m \in \mathbb{N}$ be the maximal number, such that there is an induced path $u_1,\dots,u_m$ with $u_1 = w_3$ and $u_iw_j, u_iv_l \notin E(G)$ for all $i \in [m]$, $j \in [2]$, $l \in [3]$.
 Consider some $u_{m+1} \in V(G)$, such that $u_{m+1}u_m \in E(G)$ and $u_{m+1}u_{m-1} \notin E(G)$ (if $m=1$ then $u_0 = v_3$).
 Note that such a vertex always exists by property \ref{item:forbid-net-3}.
 Since $m$ was chose to be maximal there is some vertex $v \in X := \{v_1,v_2,v_3,w_1,w_2,w_3,u_1,\dots,u_{m-1}\}$ with $u_{m+1}v \in E(G)$.
 This gives a cycle $C \subseteq X \cup \{u_m,u_{m+1}\}$ of length $l \geq 4$ and some $i \in [2]$, such that $uw_i \notin E(G)$ for all $u \in C \setminus \{u_{m+1}\}$.
 By \ref{item:forbid-net-2} there is an edge $u_{m+1}w_i \in E(G)$.
 But then the induced subgraph $G[u_{m+1},u_{m},v,w_i]$ builds a claw contradicting \ref{item:forbid-net-1}.
\end{proof}

\subsection{Neighborhood graphs}

Before considering neighborhoods of cliques we first restrict to neighborhoods of vertices.
This occurs as a subcase when analyzing neighborhood of cliques.
Also the structure of neighborhoods tends to be simpler than for neighborhoods of cliques.

\begin{definition}
 A unit square graph is a \emph{neighborhood graph} if there is a realization $f\colon V(G) \rightarrow [-1,1]^{2}$.
\end{definition}

Note that every graph induced on a neighborhood of a vertex is indeed a neighborhood graph and every neighborhood graph can be turned into the neighborhood of a vertex by adding a universal vertex located at the origin.
In this subsection let $G$ be a neighborhood unit square graph with realization $f\colon V(G) \rightarrow [-1,1]^{2}$.
The goal is to prove that, after performing the $k$-dimensional Weisfeiler-Leman algorithm for sufficiently large $k$, each color class of vertices is co-bipartite or proper circular arc.
Building on the characterization of proper circular arc graphs in terms forbidden induced subgraphs, we have to consider $C_{n+4} \cup K_1$ and $S_3 \cup K_1$.

\begin{lemma}
 \label{la:forbid-circle}
 Let $X = \{v \in V(G) \mid \exists \ell \geq 4 \,\exists w_1,\dots,w_\ell: vw_i \notin E(G) \wedge G[w_1,\dots,w_\ell] \cong C_\ell\}$.
 Then $X \neq V(G)$.
\end{lemma}

\begin{proof}
 Let $v = \argmin_{v \in V(G)}|f(v)_1|$.
 Suppose towards a contradiction that $v \in X$.
 Then there is some $\ell \geq 4$ and $w_1,\dots,w_\ell \in V(G)$, such that $vw_i \notin E(G)$ for all $i \in [\ell]$ and $w_iw_j \in E(G)$ if and only if $i-j \equiv \pm 1\mod \ell$ for all $i,j \in [\ell]$.
 Without loss of generality assume that $f(v) \in [-1,0] \times [-1,0]$.
 Let $i= \argmin_{i \in [\ell]} f(w_i)_2$.
 Since $G[w_1,\dots,w_\ell]$ is not a unit interval graph it holds that $f(w_i) \in [0,1] \times [-1,0]$, by Remark \ref{re:induced-unit-interal-by-realization}.
 Without loss of generality assume $i = 2$.
 Now consider the two neighbors $w_1$ and $w_3$.
 Note that $w_1w_3 \notin E(G)$ since $\ell \geq 4$.
 Then there is some $j \in \{1,3\}$, such that $f(w_j) \in [-1,0) \times [0,1]$.
 So in particular $f(w_j)_1 < 0$.
 Further $f(w_j)_1 + 1 \geq f(w_2)_1$ and $f(v)_1 + 1 < f(w_2)_1$.
 Altogether this means that $f(v)_1 < f(w_j)_1 < 0$ contradicting the definition of $v$.
\end{proof}

\begin{lemma}
 \label{la:forbid-s3}
 Let $X = \{v \in V(G) \mid \exists w_1,\dots,w_6: vw_i \notin E(G) \wedge G[w_1,\dots,w_6] \cong S_3\}$.
 Then $X \neq V(G)$.
\end{lemma}

\begin{proof}
 Let $v = \argmin_{v \in V(G)}|f(v)_1|$.
 Suppose towards a contradiction that $v \in X$.
 Then there are vertices $w_1,\dots,w_6 \in V(G)$, such that $vw_i \notin E(G)$ for all $i \in [k]$ and $G[w_1,\dots,w_6] \cong S_3$ as depicted in Figure \ref{fig:graph-s3}.
 Without loss of generality assume that $f(v) \in [-1,0] \times [-1,0]$.
 Let $i= \argmin_{i \in [k]} f(w_i)_2$.
 Since $G[w_1,\dots,w_6]$ is not a unit interval graph it holds that $f(w_i) \in [0,1] \times [-1,0]$, by Remark \ref{re:induced-unit-interal-by-realization}.
 
 First suppose that $i \in \{4,5,6\}$.
 Then there are $j_1, j_2 \in \{1,2,3\}$ with $j_1 \neq j_2$, $w_iw_{j_1},w_iw_{j_2} \in E(G)$ and $w_{j_1}w_{j_2} \notin E(G)$.
 Further $f(w_{j_s}) \in [-1,0] \times [0,1]$ for some $s \in [2]$.
 But then $f(v)_1 < f(w_{j_s})_1 < 0$ contradicting the definition of $v$.
 
 So consider $i \in \{1,2,3\}$.
 Without loss of generality assume $i = 1$.
 Using the same argument as before $f(w_4), f(w_6) \in [0,1] \times [-1,1]$.
 Assume $f(w_6)_1 \leq f(w_4)_1$.
 First suppose that $f(w_6) \in [0,1] \times [0,1]$.
 It holds that $f(w_2) \in [-1,1] \times [0,1]$ and $f(w_2)_1 + 1 \geq f(w_4)_1$.
 But then $\|f(w_2) - f(w_6)\|_{\infty} \leq 1$ contradicting $w_2w_6 \notin E(G)$.
 So $f(w_6) \in [0,1] \times [-1,0]$.
 Using again the same argument as before $f(w_3), f(w_5) \in [0,1] \times [-1,1]$.
 More precisely $f(w_3), f(w_5) \in [0,1] \times [0,1]$ since there is no edge to $w_1$.
 Since $w_3w_4 \notin E(G)$ it holds that $f(w_4) \in [0,1] \times [-1,0]$.
 So $f(w_2) \in [0,1] \times [-1,1]$.
 But this is now a contradiction to Remark \ref{re:induced-unit-interal-by-realization} since $S_3$ is not a unit interval graph.
\end{proof}

In order to prove the main partitioning result for neighborhood graphs we also require, that for sufficiently large $k$ the $k$-dimensional Weisfeiler-Leman algorithm identifies all interval graphs (cf. \cite{laubner10}).

\begin{corollary}
 \label{cor:wl-neighborhood}
 There is some $k \in \mathbb{N}$, such that for each neighborhood unit square graph the following holds:
 After performing $k$-dimensional Weisfeiler-Leman each equivalence class of vertices induces a graph which is co-bipartite or proper circular arc with at most four connected components.
\end{corollary}

\begin{proof}
 Choose $k$ sufficiently large and let $X \subseteq V(G)$ be an equivalence class.
 Then $G[X]$ is a neighborhood unit square graph.
 By Lemma \ref{la:unit-interval-neighborhood-existence} there is some $v \in X$, such that $N_{G[X]}[v]$ induces a unit interval graph.
 Since $k$-dimensional Weisfeiler-Leman identifies all interval graphs this is true for all $v \in X$.
 From Lemma \ref{la:forbid-circle} it follows that there exists a vertex $v \in X$, such that every induced cycle contains at least one vertex being a neighbor of $v$.
 Again by stability of the set $X$ with respect to $k$-dimensional Weisfeiler-Leman this is true for all $v \in X$ (note that the maximal length of an induced cycle is at most $8$).
 So there is no induced subgraph isomorphic to $C_{n+4} \cup K_1$.
 Combining the same argument with Lemma \ref{la:forbid-s3} we get that $G[X]$ also has no induced subgraph isomorphic to $S_3 \cup K_1$.
 Since $G[X]$ is regular there are no vertices $v,w \in X$, such that $N_{G[X]}[v] \subsetneq N_{G[X]}[w]$.
 So we can apply Lemma \ref{la:forbid-net} and obtain that there is no induced subgraph isomorphic to \texttt{net}.
 Furthermore $G[X]$ has no induced subgraph isomorphic to $\overline{T_2}$ by Corollary \ref{cor:forbidden-k23-3k2-t2}.
 Now suppose there is an induced subgraph isomorphic to $\overline{C_6}$.
 Then, by stability, every vertex is part of an induced subgraph $\overline{C_6}$ and thus, $G[X]$ is co-bipartite by Lemma \ref{la:forbid-complement-c6}.
 Otherwise $G[X]$ is proper circular arc by Proposition \ref{prop:pca-forbidden-characterization}.
 The bound on the number of components follows from the fact that $K_{1,5}$ is not a unit square graph (cf.\ Lemma \ref{la:forbid-k15}).
\end{proof}

\begin{theorem}
 \label{thm:partition-neighborhood}
 Let $G$ be a neighborhood unit square graph.
 Then one can compute in polynomial time a canonical clique-partition $\mathcal{P}$ and a canonical colored graph $H$, such that
 \begin{enumerate}
  \item $\mathcal{P} = V(H)$,
  \item $H$ is $4$-circle-bounded,
  \item $\im(\varphi) \leq \Aut(H)$ where $\varphi\colon\Aut(G) \rightarrow \Sym(\mathcal{P})$ is the natural action of $\Aut(G)$ on $\mathcal{P}$.
 \end{enumerate}
\end{theorem}

\begin{proof}
 Choose $k$ according to Corollary \ref{cor:wl-neighborhood} and let $X \subseteq V(G)$ be an equivalence class after performing $k$-dimensional Weisfeiler-Leman.
 Further let $c$ be the color of the equivalence class.
 First suppose $G[X]$ is co-bipartite.
 Let $t$ be the number of non-trivial connected components of $\overline{G[X]}$.
 Then $t \leq 2$ by Corollary \ref{cor:forbidden-k23-3k2-t2}.
 Let $Y_{j,1},Y_{j,2}$ be the unique bipartition of the $j$-th connected component, $j \in [t]$.
 Further let $Y$ be the set of isolated vertices in $\overline{G[X]}$ and $\mathcal{Y} = \{Y\}$ if $Y \neq \emptyset$ and $\mathcal{Y} = \emptyset$ otherwise.
 Define $\mathcal{P}_{X} = \{Y_{j,j'} \mid j \in [t], j' \in [2] \} \cup \mathcal{Y}$.
 Further let $H_{X} = \{\mathcal{P}_X,E(H_X),c_X\}$ with $YZ \in E(H_X)$ if there are $v \in Y, w \in Z$ with $vw \in E\left(\overline{G[X]}\right)$ and $c_X(Y) = c$.
 
 Otherwise $G[X]$ is proper circular arc according to Corollary \ref{cor:wl-neighborhood}.
 Let $X_1,\dots,X_t$ be the connected components of $G[X]$.
 Then $t \leq 4$ by Corollary \ref{cor:wl-neighborhood}.
 Let $i \in [t]$ and let $\mathcal{P}_{X,i}$ be the partition containing the equivalence classes of the connected twin relation for $G[X_i]$.
 Further let $H_{X,i}$ be the graph computed by Proposition \ref{prop:pca-circular-order} where each vertex is colored by $c$.
 Define $\mathcal{P}_X = \bigcup_{i \in [t]} \mathcal{P}_{X,i}$ and $H_X = \bigcup_{i \in [t]} H_{X,i}$.
 
 Finally let $\mathcal{P} = \bigcup_{X} \mathcal{P}_{X}$ and $H = \bigcup_{X} H_{X}$.
 It can easily be checked that $\mathcal{P}$ and $H$ have the desired properties.
\end{proof}

\subsection{Clique neighborhoods graphs}

Remember, that our goal is to compute a canonical clique-partition of a given unit square graph with singleton vertex $v_0$.
We first group the vertices according to their distance to $v_0$.
Then, for the first level of vertices which are all the neighbors of $v_0$, we use the previous theorem to compute a canonical clique-partition.
For all other levels we want to build up on the partition computed in the previous level.
More precisely, for a given clique in the partition of the previous level we want to partition its neighbors in the current level.
Hence, we need to consider neighborhoods of cliques.

Let $G$ be a colored unit square graph and let $X \subseteq V(G)$ be a clique, such that $V(G) = N[X] = \bigcup_{v \in X} N[v]$.
Further suppose there is some color $i$, such that $X = V_i(G)$, and there is some $k \in [|X|]$, such that $|N[v] \cap X| = k$ for all $v \in V(G) \setminus X$.
In this case $G$ is called a \emph{simple clique neighborhood graph with respect to $X$}.
In this subsection we aim to analyze the structure of $G[V \setminus X]$ and to obtain results similar to Theorem \ref{thm:partition-neighborhood}.
The most simple case occurs if $k = |X|$.
Then $G$ is a neighborhood unit square graph and thus, it can be handled by Theorem \ref{thm:partition-neighborhood}.
Suppose $k < |X|$.
Let $f:V(G) \rightarrow \mathbb{R}^{2}$ be a realization for $G$.
Without loss of generality assume that $f(x) \in [-\frac{1}{2},\frac{1}{2}]^{2}$ for all $x \in X$.
Then $f(v) \in [-\frac{3}{2},\frac{3}{2}]^{2} \setminus [-\frac{1}{2},\frac{1}{2}]^{2}$ for all $v \in V(G) \setminus X$.
Let $Q = \{(q_1,q_2) \in \{-1,0,1\}^{2} \mid |q_1| + |q_2| = 1\}$.
For $q \in Q$ let $P_q = [-\frac{1}{2} + q_1,\frac{1}{2} + q_1] \times [-\frac{1}{2} + q_2,\frac{1}{2} + q_2]$.

\begin{lemma}
 \label{la:clique-neighborhood-equivalent-neighbors}
 For $q \in Q$ consider $v,w \in V(G)$ with $f(v),f(w) \in P_q$. Then $N[v] \cap X = N[w] \cap X$.
\end{lemma}

\begin{proof}
 Without loss of generality consider $q = (0,-1)$ an further assume that $f(v)_2 \leq f(w)_2$.
 Let $x \in N[v] \cap X$.
 Then $f(v)_2 \leq f(x)_2 \leq f(v)_2 + 1$. So $|f(w)_2 - f(x)_2| \leq 1$.
 Further $f(w)_1,f(x)_1 \in [-\frac{1}{2},\frac{1}{2}]$.
 Hence, $x \in N[w] \cap X$.
 So $N[v] \cap X \subseteq N[w] \cap X$ and $|N[v] \cap X| = |N[w] \cap X|$, which implies $N[v] \cap X = N[w] \cap X$.
\end{proof}

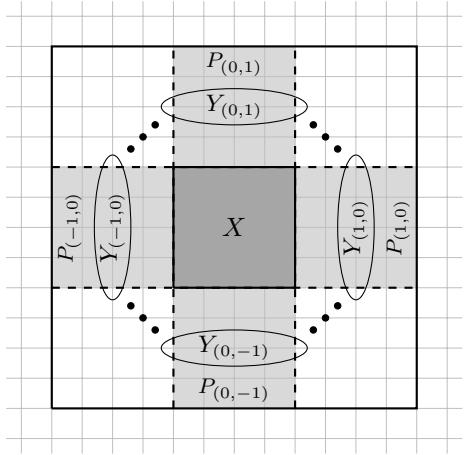
\begin{figure}
 \centering
 \begin{tikzpicture}[scale=0.8]
  \draw[fill, gray!30] (-1,1) rectangle (1,3);
  \draw[fill, gray!30] (-1,-1) rectangle (1,-3);
  \draw[fill, gray!30] (-1,1) rectangle (-3,-1);
  \draw[fill, gray!30] (1,-1) rectangle (3,1);
  \draw[step=0.5, gray!50, very thin] (-3.75,-3.75) grid (3.75,3.75);
  \draw[fill, gray!70] (-1,-1) rectangle (1,1);
  
  \draw[dashed, thick] (-3,1) -- (3,1);
  \draw[dashed, thick] (-3,-1) -- (3,-1);
  \draw[dashed, thick] (1,-3) -- (1,3);
  \draw[dashed, thick] (-1,-3) -- (-1,3);
  
  \draw[thick] (-3,-3) -- (-3,3) -- (3,3) -- (3,-3) -- (-3,-3);
  \draw[thick] (-1,-1) -- (-1,1) -- (1,1) -- (1,-1) -- (-1,-1);
  
  \draw (0,2) ellipse (1.2 and 0.3);
  \draw (0,-2) ellipse (1.2 and 0.3);
  \draw (2,0) ellipse (0.3 and 1.2);
  \draw (-2,0) ellipse (0.3 and 1.2);
  
  \draw[fill] (1.3,1.7) circle (1.5pt);
  \draw[fill] (1.5,1.5) circle (1.5pt);
  \draw[fill] (1.7,1.3) circle (1.5pt);
  
  \draw[fill] (1.3,-1.7) circle (1.5pt);
  \draw[fill] (1.5,-1.5) circle (1.5pt);
  \draw[fill] (1.7,-1.3) circle (1.5pt);
  
  \draw[fill] (-1.3,1.7) circle (1.5pt);
  \draw[fill] (-1.5,1.5) circle (1.5pt);
  \draw[fill] (-1.7,1.3) circle (1.5pt);
  
  \draw[fill] (-1.3,-1.7) circle (1.5pt);
  \draw[fill] (-1.5,-1.5) circle (1.5pt);
  \draw[fill] (-1.7,-1.3) circle (1.5pt);

  \node at (0,2) {{\small $Y_{(0,1)}$}};
  \node at (0,-2) {{\small $Y_{(0,-1)}$}};
  \node[rotate=90] at (2,0) {{\small $Y_{(1,0)}$}};
  \node[rotate=90] at (-2,0) {{\small $Y_{(-1,0)}$}};
  
  \node at (0,2.7) {{\footnotesize $P_{(0,1)}$}};
  \node at (0,-2.7) {{\footnotesize $P_{(0,-1)}$}};
  \node[rotate=90] at (2.7,0) {{\footnotesize $P_{(1,0)}$}};
  \node[rotate=90] at (-2.7,0) {{\footnotesize $P_{(-1,0)}$}};
  
  \node at (0,0) {$X$};
 \end{tikzpicture}
 \caption{Simple clique neighborhood graph with respect to $X$}
 \label{fig:neighborhood-clique}
\end{figure}

For $v,w \in V(G) \setminus X$ define $v \sim_{\texttt{snh}} w$ if $N[v] \cap X = N[w] \cap X$.
Let $\mathcal{P}_{\texttt{snh}}$ be the corresponding partition into equivalence classes.
Further for $q \in Q$ let $Y_q \in \mathcal{P}_{\texttt{snh}}$, such that $f(v) \in P_q$ for some $v \in Y_q$.
If there is no $v \in V(G)$ with $f(v) \in P_q$ then we let $Y_q = \emptyset$.

\begin{lemma}
 \label{la:clique-neighborhood-corner-structure}
 Let $Z \in \mathcal{P}_{\texttt{snh}}$, such that $G[Z]$ is not a disjoint union of cliques.
 Then there is some $q \in Q$, such that $Z = Y_q$.
\end{lemma}

\begin{proof}
 Let $Z \in \mathcal{P}_{\texttt{snh}}$, such that $Z \neq Y_q$ for all $q \in Q$.
 Then $f(v) \notin P_q$ for all $q \in Q$, $v \in Z$ by Lemma \ref{la:clique-neighborhood-equivalent-neighbors}.
 In other words, for each $v \in Z$ there is some $q^{v} \in \{-1,1\}^{2}$ with $f(v) \in P_{q^{v}}' = [-\frac{1}{2} + q_1^{v},\frac{1}{2} + q_1^{v}] \times [-\frac{1}{2} + q_2^{v},\frac{1}{2} + q_2^{v}] \setminus (\bigcup_{q \in Q} P_q)$.
 This means $vw \in E(G)$ if and only if $q^{v} = q^{w}$ for all $v,w \in Z$.
 Hence, $G[Z]$ is a disjoint union of cliques.
\end{proof}

Let $\mathcal{P}_{\texttt{snh}}^{*} = \{Z \in \mathcal{P}_{\texttt{snh}} \mid G[Z] \text{ is a disjoint union of cliques}\}$.
Further let $\mathcal{P}_{\texttt{clique}}^{*} = \{Z \subseteq V(G) \setminus X \mid \exists Z' \in \mathcal{P}_{\texttt{snh}}^{*} \colon Z \text{ is a connected component in } G[Z']\}$.
Define the graph $G_{\texttt{clique}} = (\mathcal{P}_{\texttt{clique}}^{*}, E(G_{\texttt{clique}}))$ with
\begin{equation*}
 E(G_{\texttt{clique}}) = \{YZ \mid \forall v \in Y, w \in Z \colon vw \in E(G)\}.
\end{equation*}
Further let $V^{*} = \bigcup_{X \in \mathcal{P}_{\texttt{clique}}^{*}} X$.

\begin{lemma}
 \label{la:g-clique}
 If $G[V^{*}]$ is connected then $G_{\texttt{clique}}$ is a proper circular arc graph with at most eight connected components.
\end{lemma}

\begin{proof}
 Suppose $G[V^{*}]$ is connected.
 For $q \in \{-1,1\}$ let $P_q = [-\frac{1}{2} + q_1,\frac{1}{2} + q_1] \times [-\frac{1}{2} + q_2,\frac{1}{2} + q_2] \setminus (\bigcup_{q \in Q} P_q)$.
 Let $V^{*} = \bigcup_{X \in \mathcal{P}_{\texttt{clique}}^{*}} X$.
 Let $V_q = \{v \in V^{*} \mid f(v) \in P_q\}$.
 Further let $\sim_{\texttt{clique}}$ be the equivalence relation that corresponds to $\mathcal{P}_{\texttt{clique}}^{*}$.
 \vspace*{10pt}\\
 \textit{Claim 1:} Let $t \in \{-1,1\}$ and let $p,q \in \{-1,0,1\} \times \{t\}$.
 Further let $u \in V_q, w \in V_p$ with $u \sim_{\texttt{snh}} w$.
 Then $u \sim_{\texttt{snh}} v$ for all $v \in V^{*}$ with $f(v) \in [f(u)_1,f(w)_1] \times [\frac{3}{2}t,\frac{1}{2}t]$.
 \vspace*{5pt}\\
 Without loss of generality assume $t=-1$.
 It holds that $N[u] \cap X = N[w] \cap X = N[u] \cap N[w] \cap X$.
 First consider the case that $|f(v)_2| \leq \max\{|f(u)_2|,|f(w)_2|\}$.
 Let $x \in N[u] \cap X$.
 Then $|f(u)_1 - f(x)_1| \leq 1$ and $|f(w)_1 - f(x)_1| \leq 1$, so $|f(v)_1 - f(x)_1| \leq 1$.
 Further $|f(v)_2 - f(x)_2| = |f(v)_2| - |f(x)_2| \leq \max\{|f(u)_2|,|f(w)_2|\} - |f(x)_2| = \max\{|f(u)_2 - f(x)_2|,|f(w)_2 - f(x)_2|\} \leq 1$.
 Hence, $x \in N[v] \cap X$ and $N[u] \cap X = N[v] \cap X$.
 So consider the other case, namely $|f(v)_2| \geq \max\{|f(u)_2|,|f(w)_2|\}$.
 Let $x \in N[v] \cap X$.
 Then $|f(u')_2 - f(x)_2| = |f(u')_2| - |f(x)_2| \leq |f(v)_2| - |f(x)_2| = |f(v)_2 - f(x)_2| \leq 1$ for $u' \in \{u,w\}$.
 Further there exists $u' \in \{u,w\}$ with $|f(u')_1 - f(x)_1| \leq |f(v)_1 - f(x)_1| \leq 1$ and thus, $x \in N[u'] \cap X$.
 So $N[v] \cap X = N[u] \cap X = N[w] \cap X$.
 \vspace*{10pt}\\
 Clearly the same statement also holds if the value $t$ would be in the first coordinate.
 Let $q \in \{-1,0,1\}^{2}$, $q \neq (0,0)$. For $v,w \in V_q$ define $v < w$ if $f(v)_1 < f(w)_1$.
 It follows from Lemma \ref{la:clique-neighborhood-corner-structure} and Claim 1 that for each $Z \in \mathcal{P}_{\texttt{clique}}^{*}$ the set $Z \cap V_q$ consists of consecutive elements with respect to $<$.
 For $q \in \{-1,0,1\}$, $q \neq (0,0)$, let
 \begin{equation*}
  \mathcal{V}_q = \begin{cases}
                   \{Z \cap V_q \mid Z \in \mathcal{P}_{\texttt{clique}}^{*},\;Z \cap V_q \neq \emptyset\} &\text{if } V_q \neq \emptyset\\
                   \{\emptyset\}                                                                           &\text{otherwise}
                  \end{cases}
  .
 \end{equation*} 
 Let $\mathcal{V} = \biguplus_{q \in \{-1,0,1\}^{2}, q \neq (0,0)} \mathcal{V}_q$.
 Note that $V_q \in \mathcal{V}$ for all $q \in Q$ by Lemma \ref{la:clique-neighborhood-equivalent-neighbors}.
 So $|\mathcal{V}_q| = 1$ for all $q \in Q$.
 We now define a cycle $H$ on $\mathcal{V}$ as follows.
 Let $q \in \{-1,1\}^{2}$ and let $\mathcal{V}_q = \{Z_{q,1},\dots,Z_{q,k_q}\}$ so that $f(v)_1 < f(w)_1$ for all $v \in Z_{q,i}, w \in Z_{q,j}$ and $i < j \in [k_q]$.
 Observe that such an ordering exists by Claim 1.
 Then let
 \begin{align*}
  E(H) =\; &\{Z_{q,i}Z_{q,i+1} \mid q \in \{-1,1\}^{2}, i \in [k_q - 1]\}\\
    \cup\; &\{V_{(-1,0)}Z_{(-1,q_2),1} \mid q_2 \in \{-1,1\}\} \cup \{V_{(0,q_2)}Z_{(-1,q_2),k_{(-1,q_2)}} \mid q_2 \in \{-1,1\}\}\\
    \cup\; &\{V_{(0,q_2)}Z_{(1,q_2),1} \mid q_2 \in \{-1,1\}\} \cup \{V_{(1,0)}Z_{(1,q_2),k_{(1,q_2)}} \mid q_2 \in \{-1,1\}\}.
 \end{align*}
 \vspace*{10pt}\\
 \textit{Claim 2:} Let $q \in \{-1,1\}^{2}$, $v,w \in V_q$. If $f(v)_1 \leq f(w)_1$ and $q_1q_2f(v)_2 \leq q_1q_2f(w)_2$, then $v \sim_{\texttt{snh}} w$.
 \vspace*{5pt}\\
 Without loss of generality consider $q = (-1,-1)$.
 Let $v,w \in V_q$ with $f(v)_i \leq f(w)_i$ for $i \in [2]$.
 Let $x \in N[v] \cap X$.
 Then $f(v)_i \leq f(w)_i \leq f(x)_i \leq f(v)_i + 1$ and hence, $|f(w)_i - f(x)_i| \leq 1$ for $i \in [2]$.
 So $x \in N[w] \cap X$ and thus, $N[v] \cap X = N[w] \cap X$.
 \vspace*{10pt}\\
 This has two consequences.
 First for each $q \in \{-1,1\}^{2}$ the set $\mathcal{V}_q$ is equally ordered up to reflection when we order the elements corresponding to their position in the $i$-th dimension, $i \in [2]$.
 Further this implies that if we had defined $H$ using the second dimension instead of the first (cf.\ Claim 1) then the result would have been the same.
 \vspace*{10pt}\\
 \textit{Claim 3:} Let $p \neq q \in Q$ and suppose $Y_q = Y_p \neq \emptyset$. Then $V_{p+q} \subseteq Y_q$.
 \vspace*{5pt}\\
 First note that $p$ and $q$ can not be opposite of each other since then $k = |X|$, which was excluded before.
 Without loss of generality consider $q = (0,-1)$ and $p=(-1,0)$.
 Let $u \in V_p$, $v \in V_q$ and $w \in V_{p+q}$.
 Since $u \sim_{\texttt{snh}} v$ we get $N[u] \cap X = N[v] \cap X = N[u] \cap N[v] \cap X$.
 First consider the case that $f(w)_1 \leq f(u)_1$.
 Let $x \in N[w] \cap X$.
 Then $|f(u)_1 - f(x)_1| \leq 1$ because $f(w)_1 \leq f(u)_1 \leq f(x)_1 \leq f(w)_1 + 1$.
 Further $f(u)_2 \in [-\frac{1}{2},\frac{1}{2}]$.
 So $x \in N[u] \cap X$ and thus, $N[u] \cap X = N[w] \cap X$.
 For $f(w)_2 \leq f(v)_2$ the same argument can be used to show that $N[v] \cap X = N[w] \cap X$.
 Otherwise $f(w)_1 \geq f(u)_1$ and $f(w)_2 \geq f(v)_2$.
 Let $x \in N[u] \cap N[v] \cap X$.
 Then $|f(w)_1 - f(x)_1| \leq 1$ because $f(u)_1 \leq f(w)_1 \leq f(x)_1 \leq f(u)_1 + 1$.
 Further $|f(w)_2 - f(x)_2| \leq 1$ since $f(v)_2 \leq f(w)_2 \leq f(x)_2 \leq f(v)_2 + 1$.
 So $x \in N[w] \cap X$ and hence, $N[w] \cap X = N[u] \cap X$.
 \vspace*{10pt}\\
 \textit{Claim 4:} Let $p \neq q \in Q$ and suppose $Y_q = Y_p \neq \emptyset$.
 Further assume that there are $v \in V_q, w \in V_p$ with $vw \in E(G)$. Then $V_p \cup V_q \cup V_{p+q}$ forms a clique.
 \vspace*{5pt}\\
 First observe that $V_q \cup V_p$ induces a connected subgraph of $G$.
 So $V_p \cup V_q$ forms a clique due the definition of $\mathcal{P}_{\texttt{clique}}^{*}$.
 Suppose $V_{p+q} \neq \emptyset$.
 Then $V_q \cup V_p \cup V_{p+q}$ induces a connected subgraph, since otherwise $V_{p+q}$ would be disconnected of the rest of the graph.
 So the statement follows using Claim 3.
 \vspace*{10pt}\\
 \textit{Claim 5:} Let $p,q \in Q$ with $q_i \neq p_i$ for $i \in [2]$. Further let $v \in V_q, w \in V_p$ and $u \in V_{p+q}$, such that $vw \in E(G)$.
 Then $\{u,v,w\}$ forms a clique or $u \sim_{\texttt{clique}} v$ or $u \sim_{\texttt{clique}} w$.
 \vspace*{5pt}\\
 Without loss of generality assume $q = (-1,0)$ and $p = (0,-1)$.
 First consider the case that $f(u)_1 \leq f(v)_1$.
 Let $x \in N[u] \cap X$.
 Then $f(u)_1 \leq f(v)_1 \leq f(x)_1 \leq f(u)_1 + 1$ and $f(v)_2,f(x)_2 \in [-\frac{1}{2},\frac{1}{2}]$.
 Thus, $x \in N[v] \cap X$ and $u \sim_{\texttt{snh}} v$.
 Analogously $u \sim_{\texttt{snh}} w$ if $f(u)_2 \leq f(w)_2$.
 
 If $f(u)_1 \leq f(v)_1$ and $f(u)_2 \leq f(w)_2$ then $\{u,v,w\}$ forms a clique by Claim 4.
 If $f(u)_1 \leq f(v)_1$ and $f(u)_2 \geq f(w)_2$ then $u \sim_{\texttt{clique}} v$ because $uv \in E(G)$ and $u \sim_{\texttt{snh}} v$.
 Similarly $u \sim_{\texttt{clique}} w$ if $f(u)_1 \geq f(v)_1$ and $f(u)_2 \leq f(w)_2$.
 Otherwise $\{u,v,w\}$ forms a clique, because $vw \in E(G)$, $f(v)_1 \leq f(u)_1 \leq f(w)_1$ and $f(v)_2 \geq f(u)_2 \geq f(w)_2$.
 \vspace*{10pt}\\
 \textit{Claim 6:} For each $Z \in \mathcal{P}_{\texttt{clique}}^{*}$ there is $\mathcal{V}_H(Z) \subseteq \mathcal{V}$, such that $H[\mathcal{V}_H(Z)]$ is connected and $\bigcup_{Z' \in \mathcal{V}_H(Z)} = Z$.
 \vspace*{5pt}\\
 By definition the set $Z$ induces a clique, so there is $t \in [3]$ and $q_1,\dots,q_t \in \{-1,0,1\}^{2}\setminus\{(0,0)\}$, such that $Z \subseteq \bigcup_{i \in [t]}V_{q_i}$.
 Further suppose that $t$ is minimal.
 If $t=1$ then $Z = Z \cap V_{q_1} \in \mathcal{V}_q \subseteq \mathcal{V}$.
 In this case let $\mathcal{V}_H(Z) = \{Z\}$.
 So consider $t=2$.
 First suppose there is some $i \in [2]$ with $(q_1)_i = (q_2)_i$.
 Without loss of generality assume $i=2$.
 Let $\mathcal{V}_H(Z) = \{Z \cap V_{q_1}, Z \cap V_{q_2}\}$.
 It follows from Claim 1 that these two elements are neighbors in $H$.
 Otherwise $q_1,q_2 \in Q$ and $Y_{q_1} = Y_{q_2} \neq \emptyset$.
 Further there are $v \in V_{q_1}$, $w \in V_{q_2}$ with $vw \in E(G)$.
 So $V_{q_1+q_2} = \emptyset$ and $Z = V_{q_1} \cup V_{q_2}$.
 Let $\mathcal{V}_H(Z) = \{V_{q_1}, V_{q_2}, V_{q_1+q_2}\}$, which is connected in $H$.
 Finally consider $t=3$.
 Then, without loss of generality, $q_1,q_2 \in Q$ and $q_3 = q_1 + q_2$.
 So the statement follows from Claim 3 and 4 by choosing $\mathcal{V}_H(Z) = \{V_{q_1}, V_{q_2}, V_{q_1+q_2}\}$.
 \vspace*{10pt}\\
 In the following we assume that neither endpoint of $H[\mathcal{V}_H(Z)]$ is labeled by the empty set.
 \vspace*{10pt}\\
 \textit{Claim 7:} For each $Z \in \mathcal{P}_{\texttt{clique}}^{*}$ there is a set $E \subseteq \mathcal{V}$ of elements labeled with the empty set, such that the graph $H[\bigcup_{Z' \in N_{G_{\texttt{clique}}}[Z]} \mathcal{V}_H(Z') \cup E]$ is connected.
 \vspace*{5pt}\\
 We consider two cases.
 First assume $Z \cap V_q = \emptyset$ for all $q \in Q$.
 In this case $\mathcal{V}_H(Z) = \{Z\}$.
 Then $Z \subseteq V_p$ for some $p \in \{-1,1\}^{2}$.
 Let $q_1,q_2 \in Q$ with $q_1 + q_2 = p$.
 Further let $H[\mathcal{R}_{q_1,q_2}]$ be the path from $V_{q_1}$ to $V_{q_2}$ that contains $Z$.
 Then $H[\bigcup_{Z' \in N_{G_{\texttt{clique}}}[Z]} \mathcal{V}_H(Z')]$ is a connected subgraph of $H[\mathcal{R}_{q_1,q_2}]$ by Claim 1 and 2.
 
 So let $q \in Q$ with $Z \subseteq Y_q$.
 If there are distinct $q_1,q_2 \in Q$ with $Z \cap V_{q_i} \neq \emptyset$ then it follows from Claim 3 and 4 that $N_{G_{\texttt{clique}}}[Z] = \{Z\}$, so the statement follows from Claim 6.
 So suppose there is a unique $q \in Q$ with $Z \cap V_{q} \neq \emptyset$.
 If there is some $Z' \in N_{G_{\texttt{clique}}}[Z]$ and $q \neq p \in Q$ with $Z' \cap V_p \neq \emptyset$ then $\mathcal{V}_{p+q} \subseteq \bigcup_{Z' \in N_{G_{\texttt{clique}}}[Z]} \mathcal{V}_H(Z')$ or $\mathcal{V}_{p+q} = \{\emptyset\}$ by Claim 5.
 So the statement follows in combination with Claim 1 and 2.
 \vspace*{10pt}\\
 \textit{Claim 8:} $G_{\texttt{clique}}$ is a proper circular arc graph.
 \vspace*{5pt}\\
 Let $Z_1,Z_2 \in \mathcal{P}_{\texttt{clique}}^{*}$ and let $H_i = H[\bigcup_{Z' \in N_{G_{\texttt{clique}}}[Z_i]} \mathcal{V}_H(Z') \cup E_i]$ be defined according to Claim 7.
 Consider the situation that $H_1$ is a subgraph of $H_2$ and suppose towards a contradiction that they do not share an endpoint.
 Then $Z_1 \in V(H_2)$ and let $S$ be the endpoint of $H_2$, so that $Z_1$ lies on the path from $Z_2$ to $S$ in $H_2$.
 We can assume that $S \notin E_2$.
 If $Z_2 \subseteq V_{q}$ for some $q \in \{-1,1\}^{2}$ then $Z_1 \subseteq V_q$ because $Z_1$ is not an endpoint of $H_2$.
 But then $S$ is a neighbor of $Z_1$ by Claim 1 and 2.
 So there is a unique $q \in Q$, such that $Z_2 \cap V_q \neq \emptyset$.
 Since $Z_1$ is not an endpoint of $H_2$ there is some $p \in \{-1,1\}^{2}$ with $Z_1 \subseteq V_p$.
 But then either $S \subseteq V_p$ or $S \subseteq V_{q'}$ with $q+q' = p$.
 In the first case $Z_1$ and $S$ are clearly neighbors, in the other case this follows from Claim 5.
 
 So by contradicting each $\mathcal{V}_H(Z)$ to a single vertex and also contradicting edges incident to a vertex labeled by the empty set it follows from Claim 6,7 and Proposition \ref{prop:pca-by-cycle} that $G_{\texttt{clique}}$ is a proper circular arc graph.
 \vspace*{10pt}\\
 \textit{Claim 9:} $G_{\texttt{clique}}$ has at most eight connected components.
 \vspace*{5pt}\\
 First observe that for each $q \in Q$ there is at most one $Z \in \mathcal{P}_{\texttt{clique}}^{*}$, such that $V_q \cap Z \neq \emptyset$, by Lemma \ref{la:clique-neighborhood-equivalent-neighbors}.
 But after removing these vertices the remaining vertices induce a graph, which is the disjoint union of at most $4$ cliques.
\end{proof}

\begin{lemma}
 \label{la:g-clique-3k2}
 $G_{\texttt{clique}}$ has no induced subgraph isomorphic to $\overline{3K_2}$.
\end{lemma}

\begin{proof}
 Suppose towards a contradiction that there is an induced subgraph isomorphic to $\overline{3K_2}$.
 Let $Y_1,Y_2,Y_3,Z_1,Z_2,Z_3 \in \mathcal{P}_{\texttt{clique}}^{*}$,
 such that the only non-edges in $G_{\texttt{clique}}[Y_1,Y_2,Y_3,Z_1,Z_2,Z_3]$ are $Y_iZ_i$ for $i \in [3]$.
 This means that for $i \in [3]$ there are $v_i \in Y_i,w_i \in Z_i$ with $v_iw_i \notin E(G)$.
 But then $G[v_1,v_2,v_3,w_1,w_2,w_3] \cong \overline{3K_2}$ contradicting Corollary \ref{cor:forbidden-k23-3k2-t2}.
\end{proof}

\begin{theorem}
 \label{thm:partition-clique-neighborhood}
 Let $G$ be a simple clique neighborhood graph with respect to $X \subseteq V(G)$.
 Then one can compute in polynomial time a canonical clique-partition $\mathcal{P}$ and a canonical colored graph $H$, such that
 \begin{enumerate}
  \item $\mathcal{P} \subseteq V(H)$ and $\mathcal{P}$ is $\Aut(H)$-invariant,
  \item $H$ is $8$-circle-bounded,
  \item $\im(\varphi) \leq \Aut(H)|_{\mathcal{P}}$ where $\varphi\colon\Aut(G) \rightarrow \Sym(\mathcal{P})$ is the natural action of $\Aut(G)$ on $\mathcal{P}$.
 \end{enumerate}
\end{theorem}

\begin{proof}
 For all $v,w \in V(G) \setminus X$ it holds that $|N[v] \cap X| = |N[w] \cap X|$.
 Let $k = |N[v] \cap X|$ for some $v \in V(G)$.
 If $k = |X|$ then $G$ is a neighborhood unit square graph and hence, the statement follows from Theorem \ref{thm:partition-neighborhood}.
 Otherwise initialize $\mathcal{P}$ with the empty set and $H$ with the empty graph.
 Let $\mathcal{B} = \{Y \in \mathcal{P}_{\texttt{snh}} \mid G[Y] \text{ is not a disjoint union of cliques}\}$.
 Then $|\mathcal{B}| \leq 4$ by Lemma \ref{la:clique-neighborhood-corner-structure}.
 Add the elements of $\mathcal{B}$ to the vertex set of $H$ and color each $Y \in \mathcal{B}$ by the color $\texttt{nuc-init}$.
 For each $Y \in \mathcal{B}$ the graph $G[Y]$ is a neighborhood unit square graph.
 Let $\mathcal{P}_Y$ be the partition and $H_Y = (V_Y,E_Y,c_Y)$ be the graph computed by Theorem \ref{thm:partition-neighborhood}.
 Add each element of $\mathcal{P}_Y$ to the partition $\mathcal{P}$ and add the elements of $V_Y$ to the vertex set of $H$.
 Each $v \in V_Y$ is colored by $(\texttt{nuc},c_Y(v))$ and connected to $Y$.
 Additionally all edges in $E_Y$ are added to the graph $H$.
 
 Now let $B = \bigcup_{Y \in \mathcal{B}}Y$ and $\overline{B} = V(G) \setminus (X \cup B)$.
 Let $\mathcal{Y}$ be the partition into the connected components of $G[\overline{B}]$ and add each element of $\mathcal{Y}$ to the vertex set $H$.
 Furthermore each $Y \in \mathcal{Y}$ is colored by $\texttt{cc-par}$.
 Note that $|\mathcal{Y}| \leq 8$.
 Now consider $Y \in \mathcal{Y}$.
 Then $G[Y \cup X]$ is a simple clique neighborhood unit square graph and $\mathcal{P}_{\texttt{clique}}^{*}$ partitions the set $Y$.
 Consider the graph $G_{\texttt{clique}}$, which is a proper circular arc graph with at most eight connected components according to Lemma \ref{la:g-clique}.
 Let $\mathcal{C} \subseteq \mathcal{P}_{\texttt{clique}}^{*}$ be a connected component of $G_{\texttt{clique}}$.
 Let $C = \bigcup_{A \in \mathcal{C}} A$ and add $C$ to the vertex set of $H$, where $C$ is connected to $Y$ and colored by $\texttt{clique-cc}$.
 
 If $G_{\texttt{clique}}[\mathcal{C}]$ is co-bipartite then $\overline{G_{\texttt{clique}}[\mathcal{C}]}$ has $t \leq 2$ non-trivial connected components by Lemma \ref{la:g-clique-3k2}.
 For $i \in [t]$ let $\mathcal{Z}_{i,1}, \mathcal{Z}_{i,2}$ be the unique bipartition of the $i$-th non-trivial component of $\overline{G_{\texttt{clique}}[\mathcal{C}]}$.
 Further let $\mathcal{Z}$ be the set of isolated vertices in $\overline{G_{\texttt{clique}}[\mathcal{C}]}$.
 Let $Z_{i,j} = \bigcup_{A \in \mathcal{Z}_{i,j}} A$ for $i \in [t]$, $j \in [2]$ and $Z = \bigcup_{A \in \mathcal{Z}} A$.
 Add $Z_{i,1}, Z_{i,2}$, $i \in [t]$ to the partition $\mathcal{P}$ and to the vertex set of $H$, where each of those vertices is connected to $C$ and colored by $\texttt{fin-par}$.
 Note that $G[Z_{i,j}]$ is a clique by the definition of $G_{\texttt{clique}}$.
 If $Z \neq \emptyset$ then $Z$ is also added to $\mathcal{P}$ and the vertex set of $H$, connected to $C$ and colored by $\texttt{fin-par}$.
 
 So consider the case that $G_{\texttt{clique}}[\mathcal{C}]$ is not co-bipartite.
 For each $v \in C$ let $M(v)$ be the unique set with $M(v) \in \mathcal{P}_{\texttt{clique}}^{*}$ and $v \in M(v)$.
 Let $v \sim w$ if $M(v)$ and $M(w)$ are connected twins in $G_{\texttt{clique}}$.
 Let $\mathcal{Z}$ be the corresponding partition into equivalence classes.
 Note that each $Z \in \mathcal{Z}$ induces a clique in $G$ by the definition of $G_{\texttt{clique}}$.
 Each element $Z \in \mathcal{Z}$ is added to $\mathcal{P}$ and to the vertex set of $H$.
 Further $Z$ is colored by $\texttt{fin-par}$ and we add edges according to the cycle obtained from Proposition \ref{prop:pca-circular-order}.
 This completes the description of $\mathcal{P}$ and $H$.
 
 Clearly $\mathcal{P}$ and $H$ are canonically defined and $\mathcal{P} \subseteq H$.
 Further $\mathcal{P}$ is invariant under $\Aut(H)$ because of the vertex colors.
 Finally it can easily be verified for each step of the construction that the graph $H$ is $8$-circle-bounded.
\end{proof}

\section{Global Structure}
\label{sec:global}

In this section we are ready construct a canonical, clique-stable partition $\mathcal{P}$ together with some canonical $8$-circle-bounded graph $H$,
such that $\mathcal{P} \subseteq V(H)$ and $\mathcal{P}$ is $\Aut(H)$-invariant.
This method is the central part of our algorithm and gives us a good supergroup of the natural action of the automorphism group on the computed partition.
The computed supergroup is then given to the subroutine that computes setwise stabilizers for groups in $\Gamma_8$, to obtain the automorphism group of $G[\mathcal{P}]$.
To achieve this goal we proceed in two steps.
First, we compute a clique-partition $\mathcal{P}$ which is only canonical but not necessarily clique-stable, together with a corresponding graph $H$.
For this part of the algorithm we make use of the partitioning algorithm for neighborhoods of cliques established in the previous section.
Then, in a second step, we refine the computed partition using the color refinement algorithm while simultaneously updating the graph $H$.

\subsection{Clique-Partition}

In this subsection we implement the method which realizes the first part of the described algorithm.

\begin{theorem}
 \label{thm:global-structure}
 Let $G$ be a connected unit square graph with singleton vertex.
 Then one can compute in polynomial time a canonical clique-partition $\mathcal{P}$ and a canonical colored graph $H$, such that
 \begin{enumerate}
  \item \label{item:global-structure-1} $\mathcal{P} \subseteq V(H)$ and $\mathcal{P}$ is invariant under $\Aut(H)$,
  \item \label{item:global-structure-2} $H$ is $8$-circle-bounded,
  \item \label{item:global-structure-3} $\im(\varphi) \leq \Aut(H)|_\mathcal{P}$ where $\varphi\colon\Aut(G) \rightarrow \Sym(\mathcal{P})$ is the natural action of $\Aut(G)$ on $\mathcal{P}$.
 \end{enumerate}
\end{theorem}

\begin{proof}
 Let $v_0 \in V(G)$ be a singleton vertex.
 For $i \geq 0$ let $Y_i := \{v \in V(G) \mid d(v_0,v) = i\}$.
 Further let $G_i := G[\bigcup_{j \leq i} Y_j]$.
 Then there is some $N \leq n = |V(G)|$, such that $G_N = G$.
 Let $\mathcal{P}_0 = \{\{v_0\}\}$ and $H_0$ be the graph consisting of only one vertex, namely $X_0 = \{v_0\}$.
 
 Let $i \in [N]$.
 Suppose we have already computed a canonical clique-partition $\mathcal{P}_{i-1}$ and a canonical colored graph $H_{i-1}$ with the desired properties for the graph $G_{i-1}$.
 Initialize the graph $H_i$ with $H_{i-1}$ and $\mathcal{P}_i$ with $\mathcal{P}_{i-1}$.
 For $v \in Y_{i}$ let $\mathcal{N}_{i-1}(v) = \{X \in \mathcal{P}_{i-1} \mid \exists w \in X: vw \in E(G)\}$.
 Further let $v \sim_{\texttt{init}} w$ if $\mathcal{N}_{i-1}(v) = \mathcal{N}_{i-1}(w)$.
 This gives an initial canonical partition $\mathcal{Q}_{i}$ of $Y_{i}$, which is added to the vertex set.
 Further there is an edge between $Z$ and every $X \in \mathcal{N}_{i-1}(v)$ for $v \in Z$ and each vertex $Z \in \mathcal{Q}_{i}$ is colored by $(i,\texttt{init})$.
 Consider $Z \in \mathcal{Q}_{i}$ and let $\mathcal{N}_{i-1}(v) = \{X_1^{Z},\dots,X_{k_Z}^{Z}\}$ for some $v \in Z$.
 For $v,w \in Z$ define $v \sim_{\texttt{deg}} w$ if $|N[v] \cap X_i^{Z}| = |N[w] \cap X_i^{Z}|$ for all $i \in [k_Z]$.
 This gives a canonical partition $\mathcal{Z}$ of $Z$.
 Again, the set $\mathcal{Z}$ is added to the vertex set along with vertices $(X_i^{Z},j)$ for all $i \in [k_Z]$, $j \in [|X_i^{Z}|]$.
 The vertices $(X_i^{Z},j)$ are colored by $(i,\texttt{init},j)$ and there are edges to $Z$ and $X_{i}^{Z}$.
 Each $A \in \mathcal{Z}$ is colored by $(i,\texttt{deg})$ and there is an edge to $(X_i^{Z},j)$ if $|N[v] \cap X_i^{Z}| = j$ for some $v \in A$.
 
 So consider some $A \in \mathcal{Z}$.
 Note that for each $i \in [k_Z]$ the graph $G[A \cup X_i^{Z}]$ is a simple clique neighborhood graph with respect to $X_i^{Z}$.
 For each $i \in [k_Z]$ we obtain a clique-partition $\mathcal{A}_i$ of $A$ and a graph $H_{A,i} = (V_{A,i}, E_{A,i}, c_{A,i})$ from Theorem \ref{thm:partition-clique-neighborhood}.
 For each $i \in [k_Z]$ the set $V_{A,i}$ is added to the vertex set of $H_i$.
 Each $v \in V_{A,i}$ is colored with $(i,\texttt{cnp},c_{A,i}(v))$ and connected to $A$ and $X_i^{Z}$.
 Furthermore the elements of $E_{A,i}$ are added to the edge set.
 Let $\mathcal{P}_A = \bigcap_{i \in [k_Z]} \mathcal{A}_i$.
 The set $\mathcal{P}_A$ is also added to the vertex set.
 Each $B \in \mathcal{P}_A$ is colored by $(i,\texttt{fin-par})$ and connected to the unique $A' \in \mathcal{A}_i$ with $B \subseteq A'$ for every $i \in [k_Z]$.
 Finally $\mathcal{P}_A$ is added to the partition $\mathcal{P}_i$.
 This completes the description to $H_i$ and $\mathcal{P}_i$.
 
 Clearly $\mathcal{P}_{i}$ and $H_{i}$ are canonically defined given $\mathcal{P}_{i-1}$ by Theorem \ref{thm:partition-clique-neighborhood}.
 Further $\mathcal{P}_{i}$ is a clique-partition and $\mathcal{P}_{i} \subseteq V(H_{i})$.
 From the vertex colors it is also imminent that $\mathcal{P}_{i}$ is invariant under $\Aut(H_{i})$.
 It remains to prove that $H_{i}$ is $8$-circle-bounded.
 But this can easily be verified for each layer of the construction.
 Finally note that \ref{item:global-structure-3} follows immediately from the other properties.
\end{proof}

\begin{corollary}
 \label{cor:global-structure-group}
 Let $G$ be a connected unit square graph with singleton vertex.
 Then one can compute in polynomial time a canonical clique-partition $\mathcal{P}$, such that $\im(\varphi) \in \Gamma_8$ where $\varphi\colon\Aut(G) \rightarrow \Sym(\mathcal{P})$ is the natural action of $\Aut(G)$ on $\mathcal{P}$.
\end{corollary}

\begin{proof}
 This follows directly from Theorem \ref{thm:gamma-t-t-bounded} and \ref{thm:global-structure}.
\end{proof}

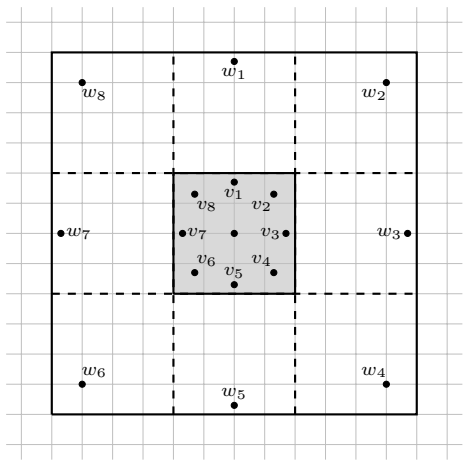
\begin{figure}
 \centering
 \begin{tikzpicture}[scale=0.8]
  \draw[fill, gray!30] (-1,-1) rectangle (1,1);
  \draw[step=0.5, gray!50, very thin] (-3.75,-3.75) grid (3.75,3.75);
  
  \draw[dashed, thick] (-3,1) -- (3,1);
  \draw[dashed, thick] (-3,-1) -- (3,-1);
  \draw[dashed, thick] (1,-3) -- (1,3);
  \draw[dashed, thick] (-1,-3) -- (-1,3);
  
  \draw[thick] (-3,-3) -- (-3,3) -- (3,3) -- (3,-3) -- (-3,-3);
  \draw[thick] (-1,-1) -- (-1,1) -- (1,1) -- (1,-1) -- (-1,-1);
  
  \draw[fill] (0,0) circle (0.05);
  
  \draw[fill] (0,-0.85) circle (0.05);
  \draw[fill] (0,0.85) circle (0.05);
  \draw[fill] (-0.85,0) circle (0.05);
  \draw[fill] (0.85,0) circle (0.05);
  \draw[fill] (-0.65,-0.65) circle (0.05);
  \draw[fill] (-0.65,0.65) circle (0.05);
  \draw[fill] (0.65,-0.65) circle (0.05);
  \draw[fill] (0.65,0.65) circle (0.05);
  
  \draw[fill] (0,-2.85) circle (0.05);
  \draw[fill] (0,2.85) circle (0.05);
  \draw[fill] (-2.85,0) circle (0.05);
  \draw[fill] (2.85,0) circle (0.05);
  \draw[fill] (-2.5,-2.5) circle (0.05);
  \draw[fill] (-2.5,2.5) circle (0.05);
  \draw[fill] (2.5,-2.5) circle (0.05);
  \draw[fill] (2.5,2.5) circle (0.05);
  
  \node at (0,0.65) {\scriptsize{$v_1$}};
  \node at (0.45,0.45) {\scriptsize{$v_2$}};
  \node at (0.6,0) {\scriptsize{$v_3$}};
  \node at (0.45,-0.45) {\scriptsize{$v_4$}};
  \node at (0,-0.65) {\scriptsize{$v_5$}};
  \node at (-0.45,-0.45) {\scriptsize{$v_6$}};
  \node at (-0.6,0) {\scriptsize{$v_7$}};
  \node at (-0.45,0.45) {\scriptsize{$v_8$}};
  
  \node at (0,2.65) {\scriptsize{$w_1$}};
  \node at (2.3,2.3) {\scriptsize{$w_2$}};
  \node at (2.55,0) {\scriptsize{$w_3$}};
  \node at (2.3,-2.3) {\scriptsize{$w_4$}};
  \node at (0,-2.65) {\scriptsize{$w_5$}};
  \node at (-2.3,-2.3) {\scriptsize{$w_6$}};
  \node at (-2.55,0) {\scriptsize{$w_7$}};
  \node at (-2.3,2.3) {\scriptsize{$w_8$}};
  
 \end{tikzpicture}
 \caption{Realization for the graph $G_8$}
 \label{fig:neighborhood-clique-graph}
\end{figure}

\begin{remark}
 The constant $d=8$ is tight for the previous corollary.
 In particular the graph $G_8$ with $V(G) = \{v_i \mid i \in [9]\} \cup \{w_i \mid i \in [8]\}$ and $E(G) = \{v_iv_j \mid i \neq j \in [9]\} \cup \{v_iw_i \mid i \in [8]\}$ is a unit square graph (the vertex $v_9$ may be a singleton vertex).
 A possible realization of $G_8$ is depicted in Figure \ref{fig:neighborhood-clique-graph}.
\end{remark}

\subsection{Refinement}

In this subsection we implement the second part of the algorithm, namely we refine the partition computed in previous section to turn it into a clique-stable partition.
This is be done by the color refinement algorithm.
Simultaneously the $8$-circle-bounded graph computed in the previous subsection is extended to obtain structure for the refined partition.
The crucial idea for extending the graph $H$ is to use additional layers which model the iterations of the color refinement algorithm.

\begin{theorem}
 \label{thm:global-structure-refined}
 Let $G$ be a connected unit square graph with singleton vertex.
 Then one can compute in polynomial time a canonical, clique-stable partition $\mathcal{P}$ and a canonical colored graph $H$, such that
 \begin{enumerate}
  \item \label{item:global-structure-refined-1} $\mathcal{P} \subseteq V(H)$ and $\mathcal{P}$ is invariant under $\Aut(H)$,
  \item \label{item:global-structure-refined-2} $H$ is $8$-circle-bounded,
  \item \label{item:global-structure-refined-3} $\im(\varphi) \leq \Aut(H)|_\mathcal{P}$ where $\varphi\colon\Aut(G) \rightarrow \Sym(\mathcal{P})$ is the natural action of $\Aut(G)$ on $\mathcal{P}$.
 \end{enumerate}
\end{theorem}

\begin{proof}
 Let $\mathcal{P}$ be the partition and $H = (V(H),E(H),c_H)$ be the graph constructed by Theorem \ref{thm:global-structure}.
 Define $\mathcal{P}_0 = \mathcal{P} \cup \{\mathcal{M}(G)\}$.
 Let $\sim_0$ be the corresponding equivalence relation.
 For $v,w \in V(G) \cup \mathcal{M}(G)$ let $v \sim_{i+1} w$ if $v \sim_i w$ and $|N_{G_{\mathcal{M}}^{*}}(v) \cap X| = |N_{G_{\mathcal{M}}^{*}}(w) \cap X|$ for all $X \in \mathcal{P}_i$.
 Further let $\mathcal{P}_{i+1}$ be the partition into the equivalence classes of $\sim_{i+1}$.
 Observe that $\mathcal{P}_i$ is the partition computed in the $i$-th round of the color refinement algorithm applied to the graph $G_{\mathcal{M}}^{*}$ if the initial partition is $\mathcal{P}_0$.
 Let $N \in \mathbb{N}$ be minimal, such that $\mathcal{P}_N = \mathcal{P}_{N+1}$.
 Then $\mathcal{P}^{*} = \mathcal{P}_{N} \cap 2^{V(G)}$ is a canonical, clique-stable partition of the vertices.
 Now define the graph $H^{*}$ with vertex set
 \begin{equation*}
  V(H^{*}) = V(H) \uplus \left(\bigcup_{0 \leq i \leq N} \left(\mathcal{P}_i \times \{i\}\right) \;\cup\; \left(\mathcal{P}_i \times \{i\} \times \{0,\dots,|V(G)|\}\right) \right) \uplus \mathcal{P}^{*}
 \end{equation*}
 as follows.
 Each vertex $v \in V(H)$ is colored the same way as in $H$, that is $c_{H^{*}}(v) = c_H(v)$.
 A vertex $v \in \mathcal{P}_i \times \{i\}$ is colored by $c_{H^{*}}(v) = (\texttt{cs-par},i)$
 and each $v \in \mathcal{P}_i \times \{i\} \times \{0,\dots,|V(G)|\}$ is colored by $c_{H^{*}}(v) = (\texttt{cs-par},i,v_3)$, where $v_3$ denotes the third component of $v$.
 The vertices $v \in \mathcal{P}^{*}$ are colored by $c_{H^{*}}(v) = \texttt{fin-cs-par}$.
 For the edge set, first each edge of $H$ is also added to $H^{*}$ and for each $X \in \mathcal{P}$ there is an edge between $X$ and $(X,0)$.
 For $i \in \{0,\dots,N\}$ the following edges are added to the graph $H^{*}$.
 Let $X \in \mathcal{P}_i$.
 Then there are edges between $(X,i)$ and $(X,i,j)$ for all $j \in \{0,\dots,|V(G)|\}$.
 Suppose $i \geq 1$.
 Then there is an edge between $(X,i)$ and $(X',i-1,j)$ if $|N(v) \cap X'| = j$ for some and thus for all $v \in X$.
 Furthermore there is an edge between $(X,i)$ and $(X',i-1)$ for the unique $X' \in \mathcal{P}_{i-1}$, such that $X \subseteq X'$.
 Finally there are edges between $(X,N)$ and $X$ for each $X \in \mathcal{P}^{*}$.
 This finishes the description of the graph $H^{*}$.
 
 Clearly $H^{*}$ and is canonically defined and $\mathcal{P}^{*} \subseteq V(H^{*})$.
 Further $\mathcal{P}^{*}$ is invariant under $\Aut(H^{*})$ because of the vertex colors.
 It remains to show that the graph $H^{*}$ is $8$-circle-bounded.
 First $H$ is $8$-circle-bounded.
 The remaining colors are ordered lexicographically, then it follows directly from the definition of $H^{*}$ that the graph is $8$-circle-bounded.
\end{proof}

\begin{remark}
 Observe, that the presented algorithm is not restricted to unit square graphs.
 In fact, whenever one is given a canonical partition and a $\Gamma_d$-group, which builds an upper bound on the top action, one can efficiently compute a $\Gamma_d$-group, which gives an upper bound on the top action of the refined partition.
\end{remark}

\subsection{Global automorphism group}

Together with Theorem \ref{thm:global-to-local} this gives us sufficient structure to compute the natural action of the automorphism group on the computed partition.
This can also be used to solve the isomorphism problem.

\begin{theorem}
 Let $G$ be a connected, twin-free unit square graph with a singleton vertex.
 Then one can compute in polynomial time a canonical, clique-stable partition $\mathcal{P}$ and a set $S \subseteq \Sym(\mathcal{P})$,
 such that $\langle S \rangle = \im(\varphi) \in \Gamma_8$ where $\varphi \colon \Aut(G) \rightarrow \Sym(\mathcal{P})$ is the natural action of $\Aut(G)$ on $\mathcal{P}$.
\end{theorem}

\begin{proof}
 Let $\mathcal{P}$ be the canonical, clique-stable partition and $H$ the canonical, $8$-circle-bounded graph obtained from Theorem \ref{thm:global-structure-refined}.
 Then $\Aut(H)$ can be computed in polynomial time and $\Aut(H) \in \Gamma_8$ by Theorem \ref{thm:automorphism-t-bounded} and \ref{thm:gamma-t-t-bounded}.
 Further $\mathcal{P}$ is invariant under $\Aut(H)$.
 Since $H$ is canonical this implies $\im(\varphi) \leq \Aut(H)|_{\mathcal{P}} \in \Gamma_8$.
 Furthermore $\im(\varphi) = \Aut(G[\mathcal{P}])$ by Theorem \ref{thm:global-to-local}.
 A generating set for $\Aut(G[\mathcal{P}])$ can be computed in polynomial time using Proposition \ref{prop:setwise-stab-gamma-d}.
\end{proof}

\begin{theorem}
 The Graph Isomorphism Problem for unit square graphs can be solved in polynomial time.
\end{theorem}

\begin{proof}
 Let $G_1,G_2$ be two unit square graphs.
 First, it can be assumed that $G_1$ and $G_2$ are connected by considering the connected components separately.
 Furthermore, the graphs can be assumed to be twin-free using modular decompositions of graphs (cf.\ \cite{schw15}).
 Let $c \in \mathbb{N}$ be a fresh color (i.e.\ a color which does not appear in $G_1$ or $G_2$).
 For a graph $G$ and a vertex $v \in V(G)$ we denote by $G^{v \mapsto c}$ the graph where vertex $v$ is colored by $c$.
 Pick $v_1 \in V(G_1)$.
 For each $v_2 \in V(G_2)$ test whether $G_1^{v_1 \mapsto c} \cong G_2^{v_2 \mapsto c}$ by the following procedure.
 For $i \in [2]$ let $\mathcal{P}_i$ be the partition and $H_i$ be the graph computed by Theorem \ref{thm:global-structure-refined} for the graph $G_i^{v_i \mapsto c}$.
 Let $H$ be the disjoint union of $H_1$ and $H_2$.
 Note that $H_1 \cong H_2$ if $G_1^{v_1 \mapsto c} \cong G_2^{v_2 \mapsto c}$ because the graph $H_i$ is canonical.
 Compute a generating set for $\Aut(H) \in \Gamma_8$.
 This can be done in polynomial time according to Theorem \ref{thm:automorphism-t-bounded}.
 Let $G$ be the disjoint union of $G_1^{v_1 \mapsto c}[\mathcal{P}_1]$ and $G_2^{v_2 \mapsto c}[\mathcal{P}_2]$.
 Then $\Aut(G) \leq \Aut(H)|_{\mathcal{P}_1 \cup \mathcal{P}_2}$ and hence a generating set for $\Aut(G)$ can be computed in polynomial time using Proposition \ref{prop:setwise-stab-gamma-d} (note that $\Aut(G)$ is the set of permutations which stabilize the edge set).
 By Theorem \ref{thm:global-to-local} it holds that $G_1^{v_1 \mapsto c} \cong G_2^{v_2 \mapsto c}$ if and only if there is an automorphism $\gamma \in \Aut(G)$ that maps $G_1^{v_1 \mapsto c}[\mathcal{P}_1]$ to $G_2^{v_2 \mapsto c}[\mathcal{P}_2]$.
 Since $G$ is the disjoint union of of $G_1^{v_1 \mapsto c}[\mathcal{P}_1]$ and $G_2^{v_2 \mapsto c}[\mathcal{P}_2]$ and both of these graphs are connected it holds that if such an automorphism exists then there will also be one present in the generating set of $\Aut(G)$.
 Thus it can be checked in polynomial time whether $G_1^{v_1 \mapsto c} \cong G_2^{v_2 \mapsto c}$.
\end{proof}

\begin{remark}
 The running time of the presented algorithm is dominated by the running time for the subroutine computing setwise stabilizers for groups in $\Gamma_8$, which in turn depends on the maximal size of primitive $\Gamma_8$-groups.
\end{remark}

The latter was analyzed by Babai, Cameron and Pálfy in \cite{BCP82} and proven to be polynomially bounded in the size of the permutation domain.
For a complexity analysis of the setwise stabilizer subroutine we refer to \cite{luks82, luks91, bkl83}.
Note that the setwise stabilizer subroutine is also used for computing the automorphism group of $H$ and the graph $H$ might be much larger than the original graph $G$.

\begin{remark}
 The presented algorithm also gives us some insight about the structure of the automorphism group of a unit square graph in case the graph has a singleton vertex.
 For the automorphism group, there is an invariant clique-partition, such that the natural action forms a $\Gamma_8$-group.
\end{remark}

An interesting question is whether a similar statement still holds if the given graph does not have a singleton vertex.
We leave this question open.

\section{Discussion}
\label{sec:discussion}

We presented a polynomial time algorithm solving the Graph Isomorphism Problem for unit square graphs.
The main idea of the algorithm is to canonically extract a clique-partition and graph structure that can be described by $8$-circle-bounded graphs.
This gives an upper bound on the action of the automorphism group on the canonical partition.
Then Luks' algorithm is used to precisely determine this action which gives sufficient information to decide whether two given unit square graphs are isomorphic.
So overall the presented algorithm heavily depends on group theoretic methods.
This raises the question whether the problem can also be solved without the use of such methods.
In fact, it might be that the $k$-dimensional Weisfeiler-Leman algorithm can identify every unit square graph for sufficiently large $k$.
This is left as an open question.

Furthermore it is an interesting question whether the methods presented in this work can be adapted to other geometric classes for which the isomorphism problem is still open.
At first glance a natural candidate seems to be the class of unit disk graphs.
However, it turns out that there are some crucial structural differences to unit square graphs.
In particular, large symmetric groups in the automorphism group of a unit disk graph do not necessarily originate from cliques.
In other words, there are unit disk graphs with singleton vertex, such that for each canonical clique-partition the natural action of the automorphism group contains a large symmetric group.
Another candidate is the class of unit grid intersection graphs. 
Unit grid intersection graphs can be seen as a bipartite version of unit square graphs.
However, in order to apply the methods of this work, one would require a different notion of locality since vertices being close to each other in the realization may not be connected in the graph.

Finally we would like to address two natural generalizations of unit square graphs.
The first one concerns the dimension of the realization, that is, what is the complexity of graph isomorphism for graphs with $d$-dimensional $L_{\infty}$-realization for any constant number $d$.
Intuitively, it seems that similar arguments might be applicable to higher dimensions.
In particular, for the neighborhood of a clique the number of independent neighbors is still bounded.
The second extension concerns squares of arbitrary size.
This is still a natural restriction for the class of intersection graphs of rectangles which is GI-complete.
The reduction uses complete bipartite graphs of arbitrary size as induced subgraphs (see \cite{ueh08}) which can not be modeled as intersection graphs of squares.
So it might be possible that isomorphisms test can be efficiently performed even for square graphs.
However, this would require some new ideas since the number of independent neighbors of a vertex is unbounded.

\paragraph*{Acknowledgements}

I want to thank Martin Grohe and Pascal Schweitzer for several helpful discussions and comments throughout this work.

\newcommand{\etalchar}[1]{$^{#1}$}

\end{document}